\documentclass[journal]{IEEEtran}
\ifCLASSINFOpdf
\else
\fi
\usepackage{url}
\usepackage{graphicx, color}
\usepackage{amsmath,cite,mathtools}
\usepackage{multirow}
\allowdisplaybreaks[4]
\usepackage{enumitem}

\usepackage{amsfonts,amssymb,amsthm,algorithm,algpseudocode,algorithmicx}
\alglanguage{pseudocode}
\algnewcommand{\Inputs}[1]{%
  \State \textbf{Inputs:} #1
}
\usepackage{bm}
\theoremstyle{definition}
\newtheorem{lemma}{Lemma}
\newtheorem*{lemma*}{Lemma}

\newtheorem{prop*}{Proposition*}
\newtheorem{definition}[lemma]{Definition}
\newtheorem{definition*}{Definition*}

\newtheorem{note*}{Note*}
\newtheorem{theorem}[lemma]{Theorem}
\newtheorem{theorem*}{Theorem*}

\newtheorem{col*}{Corollary*}
\def\ddiag{\mathop{\rm ddiag}\nolimits}
\def\diag{\mathop{\rm diag}\nolimits}
\def\tr{\mathop{\rm tr}\nolimits}

\title{ 
A Joint Diagonalization Based Efficient Approach to Underdetermined Blind Audio Source Separation Using the Multichannel Wiener Filter
} 
\author{Nobutaka~Ito,~\IEEEmembership{Senior~Member,~IEEE},
~Rintaro~Ikeshita,~\IEEEmembership{Member,~IEEE},%
~Hiroshi~Sawada,
and Tomohiro~Nakatani,~\IEEEmembership{Fellow,~IEEE}
\thanks{N. Ito, R. Ikeshita, H. Sawada, and T. Nakatani are with  Communication Science Laboratories,
NTT Corporation, Kyoto, Japan; email: nobutaka.ito@ieee.org.}
}
\begin{document}
\maketitle

\begin{abstract}
This paper presents a computationally efficient approach to
blind source separation (BSS) of audio signals,
applicable even when there are more sources than microphones (\textbf{\textit{i.e.,}} the underdetermined case).
When there are as many sources as microphones (\textbf{\textit{i.e.,}} the determined case),
BSS can be performed computationally efficiently by
independent component analysis (ICA). Unfortunately, however, ICA
is basically inapplicable to the underdetermined case.
Another BSS approach using the multichannel Wiener filter (MWF)
is applicable even to this case, and encompasses
 full-rank spatial covariance analysis (FCA) and multichannel non-negative matrix factorization (MNMF).
However, these methods require
massive numbers of matrix inversions to design the MWF,
and  are  thus computationally inefficient.
To resolve this challenge,
we exploit
the well-known property of diagonal matrices that matrix inversion amounts to mere inversion of the diagonal elements and can thus be performed computationally efficiently.
This makes it possible to drastically reduce the computational cost of the above matrix inversions 
based on a joint diagonalization (JD) idea, leading to computationally efficient BSS.
Specifically, we restrict the \textbf{\textit{N}} spatial covariance matrices (SCMs) of all \textbf{\textit{N}} sources to
a class of (exactly) jointly diagonalizable matrices.
Based on this approach, we present \textbf{\textit{FastFCA}}, a computationally efficient extension of FCA.
We also present a unified framework for underdetermined and determined audio BSS, which highlights a theoretical connection between FastFCA and other methods. 
Moreover, we reveal that FastFCA can be regarded as a regularized version of approximate joint diagonalization (AJD).
\end{abstract}

\section{Introduction}


\IEEEPARstart{T}{his} paper deals with blind source separation (BSS) of audio signals, 
which aims to separate source signals from their mixtures recorded by microphones. The applications include automatic speech recognition in a noisy/multi-speaker environment, hearing aids, and music analysis. 

%

BSS approaches can be broadly categorized into single-channel and multichannel approaches, where representative of the former is non-negative matrix factorization (NMF)~\cite{Lee1999NMF,Smaragdis2003,Virtanen2007,Fevotte2009}. In this paper, we focus on the multichannel approach, which can leverage spatial information contained in multichannel data. Another significant distinction is whether training data are utilized or not. Recently, many methods have been proposed based on the deep neural network (DNN) with training data~\cite{Hershey2016,Nugraha2016,Kolbaek2017}. Here we focus on BSS methods that require no training data.

A popular approach to audio BSS is based on independent component analysis (ICA)~\cite{Hyvarinen2001,Cichocki2002}
applicable to the determined case $N=M$,
where
$N$ is the number of sources and $M$ that of microphones.
In ICA, source separation is performed by using a separation matrix,
which can be estimated by exploiting
 statistical independence of the source signals. 
For example, Pham and Cardoso~\cite{Pham2001} proposed {\it time-varying Gaussian ICA}, 
where the matrix is estimated by the maximum 
likelihood method on the assumption that the source signals follow time-varying Gaussian distributions independently.
In audio BSS, ICA typically operates in the short-time Fourier transform (STFT) domain, where individual frequency components are processed independently.
In this case, ICA cannot determine by itself which separated frequency component corresponds to which source, which is called a {\it permutation ambiguity}.
It should be resolved by, {\it e.g.},
post-processing~\cite{Sawada2004TASLP,Sawada2007ISCAS}, so that
 separated frequency components originating from the same source are grouped together.
There also exist permutation-free, full-band extensions of ICA, such as independent vector analysis (IVA)~\cite{Hiroe2006,Kim2007,Ono2012APSIPA} and {\it independent low-rank matrix analysis (ILRMA)}~\cite{Kitamura2016},
where all frequency components are processed jointly. 
These ICA-based methods can perform BSS effectively and computationally efficiently in the determined case, or in the overdetermined case $N<M$ when
combined with dimensionality reduction by principal component analysis (PCA).
However, 
these methods are basically inapplicable to the underdetermined case $N>M$.


Therefore, underdetermined BSS has been recognized as a significant challenge
with
a great deal of effort devoted to it in the past two decades~\cite{Yilmaz2004,Araki2007SigPro,Ito2016EUSIPCO,Cardoso2002,Fevotte2005,Ozerov2010,Duong2010,Yoshii2013,Arberet2010,Sawada2012ICASSP,Sawada2013,Makino2007,Makino2018}.
The promising approaches are based on either
time-frequency masks~\cite{Yilmaz2004,Araki2007SigPro,Ito2016EUSIPCO}
or the MWF~\cite{Cardoso2002,Fevotte2005,Ozerov2010,Duong2010,Yoshii2013,Arberet2010,Sawada2012ICASSP,Sawada2013}.
The former approach assumes 
 that the source signals are sparse 
and thus rarely overlap in the STFT domain.
However, this assumption is not always valid: For example, in music analysis,
instrumental sounds overlap significantly to create rhythm and harmony.
In contrast, the
MWF based approach
is free of the above assumption, and is the focus of this paper.


There have been
many BSS methods proposed
based on the MWF.
Cardoso {\it et al.}~\cite{Cardoso2002,Fevotte2005} proposed an underdetermined extension of  time-varying Gaussian ICA, 
which we call {\it rank-1 spatial covariance analysis (R1CA)}.
This method employs the same
source signal model as time-varying Gaussian ICA,
but estimates a mixing matrix modeling
 the original mixing system  instead of 
the separation matrix modeling the inverse system.
This makes the method applicable even to the underdetermined case,
in which the separation matrix is not well-defined.
Ozerov and F\'evotte~\cite{Ozerov2010} developed a permutation-free, full-band extension of R1CA by
 incorporating NMF as a source signal model.
We call it
{\it rank-1 multichannel NMF (MNMF)}.
Duong, Vincent, and Gribonval
extended R1CA in another direction, yielding what we
call {\it full-rank spatial covariance analysis (FCA)}~\cite{Duong2010,Yoshii2013}.
FCA models the mixing system by {\it full-rank spatial covariance matrices (SCMs)},
whereas R1CA models it by the mixing matrix or, equivalently, rank-1 SCMs.
With the rank-1 constraint on the SCMs removed, FCA enables more flexible signal modeling  than R1CA, leading to empirically more effective BSS under reverberation~\cite{Duong2010}.
There is also a full-band extension of 
FCA, which we call {\it full-rank MNMF}~\cite{Arberet2010,Sawada2012ICASSP,Sawada2013}.
The above methods~\cite{Cardoso2002,Fevotte2005,Ozerov2010,Duong2010,Yoshii2013,Arberet2010,Sawada2012ICASSP,Sawada2013} share the drawback of computational inefficiency,
which is caused by
massive numbers of matrix inversions required for designing the MWF.
This has constituted a major obstacle to applying this approach widely in the real world,
and developing its computationally efficient extension has long been an extremely important but difficult challenge.



\subsection{Contributions of This Paper}
The contributions of this paper are three-fold:
\begin{enumerate}
\item FastFCA for efficient underdetermined BSS,
\item a unified BSS framework,
\item a relation to approximate joint diagonalization (AJD).
\end{enumerate}
Although 1) has also appeared in our preliminary work~\cite{Ito2018IWAENC,Ito2018GlobalSIP}, 
here we give a complete and definitive account of it. 
2) and 3) constitute the main novelty of this paper compared to~\cite{Ito2018IWAENC,Ito2018GlobalSIP}.
We explain each contribution in the following.

\subsubsection{FastFCA for Efficient Underdetermined BSS}
We present a computationally efficient approach to underdetermined audio BSS using the MWF.
To resolve the above challenge,
we exploit
the well-known property of diagonal matrices that matrix inversion amounts to mere inversion of the diagonal elements and can thus be performed computationally efficiently.
This makes it possible to drastically reduce the computational cost of the above matrix inversions 
based on a joint diagonalization idea, leading to computationally efficient BSS.
Specifically, we impose on the $N$ SCMs of all $N$ sources
 the constraint that they are jointly diagonalizable, which we call a  {\it joint diagonalizability (JD) constraint.}
Based on this approach, we present {\textit{FastFCA}},  a computationally efficient extension of FCA.
Interestingly, the FastFCA  cost turns out to be a `mixture' of ICA and NMF costs. 
Based on this observation, we present
 efficient optimization algorithms for FastFCA, which leverage efficient optimization algorithms that have been developed
for  ICA and NMF.
The computational efficiency of our approach makes it suitable for large data ({\textit{e.g.}}, data augmentation for machine learning) 
or limited computational resources encountered in, {\it{e.g.,}} hearing aids, distributed microphone arrays, and online processing.

\subsubsection{Unified BSS Framework}
We also present a unified framework for underdetermined and determined audio BSS, which
highlights a close theoretical connection between FastFCA and other methods.
We show that these methods share the same likelihood function, and mainly differ
in the way the source covariance matrices are parametrized. Furthermore, as shown in Fig.~\ref{fig:framework},
the connection between methods can be represented by a cube, where
the eight vertices represent eight methods and
the three axes three constraints imposed on the source covariance matrices.

\begin{figure}
\centering
\includegraphics[width=\columnwidth]{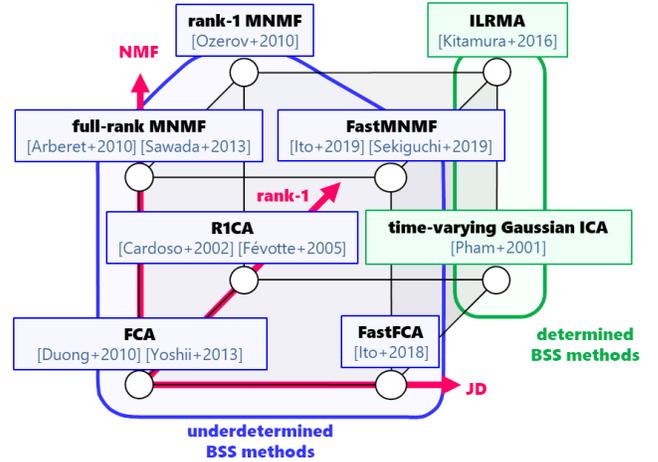}\\
\caption{A unified BSS framework, encompassing  six underdetermined and two determined BSS methods. It is represented by a cube, where the eight vertices represent the eight methods and the three axes three constraints imposed on the covariance matrices of the source images.
}
\label{fig:framework}
\end{figure}

\subsubsection{Relation to AJD}
A concept similar to, but different from, the JD constraint is approximate joint diagonalization (AJD)~\cite{Flury1986,Cardoso1993,Belouchrani1997,Pham2001,Veen2001,Yeredor2002,Alyani2017}. It seeks for a single matrix that jointly diagonalizes a given set of matrices as well as possible. 
In the statistical literature, Flury and Gautschi~\cite{Flury1986} defined the AJD problem for the first time, motivated by their study of common principal component analysis~\cite{Flury1984}. Here, 
deviation from diagonality was measured by using the log-determinant divergence\footnote{The log-determinant divergence
is also known as Stein's loss and a Burg matrix divergence. It is 
 an instance of the matrix Bregman divergence.}~\cite{James1961,Kulis2009},
a matrix extension of
 the Itakura-Saito divergence.
 Cardoso and Souloumiac~\cite{Cardoso1993} formulated separation matrix estimation in ICA as AJD of fourth-order cumulant matrices, which was called JADE. They used a least squares criterion. 
Belouchrani {\it et al.}~\cite{Belouchrani1997} showed another application of AJD to ICA, called SOBI, where they used lagged sample covariance
matrices with different lags instead of the fourth-order cumulant matrices. 
Pham and Cardoso~\cite{Pham2001} showed that
a piecewise stationary extension of time-varying Gaussian ICA
boils down to AJD of short-term sample
covariance matrices in different time intervals (hereafter called {\it Pham's AJD}). A problem related to AJD, called subspace fitting, was studied in~\cite{Veen2001,Yeredor2002}.
Here we reveal that FastFCA can be regarded as a regularized version of Pham's AJD
 with the regularization term being
the cost of NMF using the
 Itakura-Saito divergence~\cite{Fevotte2009}.

\subsection{Paper Structure}
The rest of this paper is structured as follows. 
Section~\ref{sec:form} formulates our BSS problem.
Section~\ref{sec:FCA} reviews FCA.
FastFCA is presented in Section~\ref{sec:prop}.
Section~\ref{sec:relICA} describes the unified framework, and 
Section~\ref{sec:conajd} the connection to AJD.
Section~\ref{sec:exp} is devoted to experimental evaluation, and
Section~\ref{sec:conc} concludes the paper.

%

\section{BSS Problem Formulation}
\label{sec:form}

\begin{figure}
\centering
\includegraphics[width=\columnwidth]{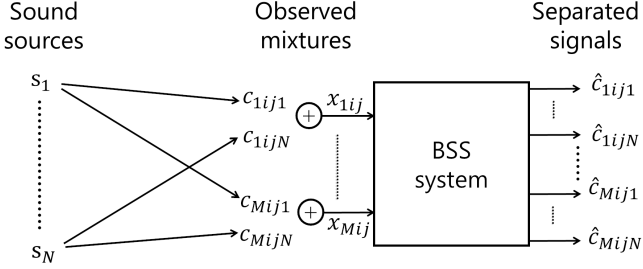}\\
\caption{Signal notations.}
\label{fig:notations}
\end{figure}

Let us formulate our BSS problem.
As in Fig.~\ref{fig:notations}, suppose $N\ (\geq 2)$ source signals emanating from $N$ sound sources are mixed, and observed by $M\ (\geq 2)$ microphones.
$N$ is assumed to be given and
 possibly larger than $M$.
We consider frequency-domain BSS, where the mixtures are transformed from the time to the frequency domains by
STFT. We denote the microphone index by $m=1,\dots,M$,
the frequency bin index by $i=1,\dots,I$,
the frame index by $j=1,\dots,J$,
and the source index by $n=1,\dots,N$.
The mixture observed by the $m$th microphone, ${x}_{mij}\in\mathbb{C}$, is modeled by the sum of all $N$ source signals as
\begin{align}
{x}_{mij}=\sum_{n=1}^N{c}_{mijn},\label{eq:scalarform}
\end{align}
where $c_{mijn}$ denotes the unknown contribution of the $n$th source signal to the $m$th microphone.
We can rewrite (\ref{eq:scalarform}) in vector form as
\begin{align}
{\mathbf{x}}_{ij}=\sum_{n=1}^N{\mathbf{c}}_{ijn},\label{eq:summodel}
\end{align}
where
\begin{align}
&{\mathbf{x}}_{ij}\coloneqq\begin{pmatrix}
{x}_{1ij}\\
 \vdots\\ {x}_{Mij}\end{pmatrix}\in\mathbb{C}^M,
&{\mathbf{c}}_{ijn}\coloneqq
\begin{pmatrix}{c}_{1ijn}\\ \vdots\\ {c}_{Mijn}\end{pmatrix}\in\mathbb{C}^M
\end{align}
are
$M$-dimensional vectors with elements corresponding to all $M$ microphones. We call ${\mathbf{c}}_{ijn}$ a {\it source image}.
Our BSS problem is that of estimating 
 ${\mathbf{c}}_{1:I,1:J,1:N}$ given ${\mathbf{x}}_{1:I,1:J}$ and $N$,
where $1:J$ (or the like) is a shorthand notation for $1,\ldots,J$ (or the like).

\section{Review: Full-rank Spatial Covariance Analysis (FCA)}
\label{sec:FCA}
This section reviews FCA~\cite{Duong2010}, an underdetermined BSS method based on the MWF.
Since FCA processes individual frequency components independently, we omit the frequency bin index $i$ for brevity in this section.

\subsection{Stochastic Signal Model}
FCA assumes that the $N$ source images independently follow zero-mean proper complex Gaussian distributions
\begin{align}
p(\mathbf{c}_{jn})\label{eq:srcstomodel}
&=\mathcal{N}_c(\mathbf{c}_{jn}\mid\mathbf{0},h_{jn}\mathbf{R}_{n}).
\end{align}
Here, $h_{jn}>0$ models the power spectrum of the $n$th source signal,
 $\mathbf{R}_n\in S_{++}^M$ is a full-rank   SCM modeling the mixing system,
$
\mathcal{N}_c(\bm{\alpha}\mid\bm{\mu},\bm{\Omega})
\coloneqq \exp[-(\bm{\alpha}-\bm{\mu})^H\bm{\Omega}^{-1}(\bm{\alpha}-\bm{\mu})]/\det(\pi\bm{\Omega})
$
denotes the probability density function of the proper complex Gaussian distribution
with mean $\bm{\mu}$ and covariance matrix $\bm{\Omega}$, 
$S_{++}^M$ the set of $M$th-order Hermitian positive definite matrices, and 
${}^H$ Hermitian transposition.
From the above assumption of source independence,
\begin{align}
\mathbb{E}[\mathbf{c}_{jn}\mathbf{c}_{j\nu}^H]=\mathbf{O}\ (n\neq \nu),\label{eq:srccross}
\end{align}
where $\mathbb{E}$ denotes expectation and $\mathbf{O}$ the zero matrix of the appropriate size.
Consequently, the mixtures have a zero-mean proper complex Gaussian distribution again:
\begin{align}
p(\mathbf{x}_{j}\mid\Theta)
&=\mathcal{N}_c(\mathbf{x}_{j}\mid\mathbf{0},\mathbf{X}_j(\Theta)),\label{eq:obsden}
\end{align}
where $\mathbf{X}_j(\Theta)\coloneqq\sum_{n=1}^Nh_{jn}\mathbf{R}_{n}$ is the observation covariance matrix and
$\Theta\coloneqq\{\mathbf{R}_{1:N},h_{1:J,1:N}\}$ the set of model parameters.

\subsection{Multichannel Wiener Filter}
In FCA, source separation is performed by using the MWF, applicable even in the underdetermined case.
The MWF is a multichannel extension of the well-known (single-channel) Wiener filter. 
Consider an estimator of ${\mathbf{c}}_{jn}$ of form
$
\widehat{\mathbf{c}}_{jn}=\mathbf{F}_{jn}\mathbf{x}_j,
$ which is a linear function of $\mathbf{x}_j$. Here, the filter coefficient matrix $\mathbf{F}_{jn}\in\mathbb{C}^{M\times M}$ is assumed to be deterministic. Different choices of $\mathbf{F}_{jn}$ lead to different estimators, but we aim to find the one that minimizes the mean square error (MSE) between the unknown source image $\mathbf{c}_{jn}$ and its estimate $\widehat{\mathbf{c}}_{jn}$, given by $MSE(\mathbf{F}_{jn})\coloneqq \mathbb{E}[\|\widehat{\mathbf{c}}_{jn}-\mathbf{c}_{jn}\|_2^2]$.
On the assumption that $\mathbb{E}\bigl[\mathbf{x}_j\mathbf{x}_j^H\bigr]$ is nonsingular, the unique minimizer is given by 
$
\arg\min_\mathbf{F} MSE(\mathbf{F})=\mathbf{F}_{jn}^{MWF}\coloneqq \mathbb{E}[\mathbf{c}_{jn}\mathbf{x}_j^H]\mathbb{E}[\mathbf{x}_j\mathbf{x}_j^H]^{-1}
$,
where $\mathbb{E}[\mathbf{c}_{jn}\mathbf{x}_j^H]$ is the cross-correlation matrix of $\mathbf{c}_{jn}$ and $\mathbf{x}_j$ and $\mathbb{E}[\mathbf{x}_j\mathbf{x}_j^H]$ the auto-correlation matrix of $\mathbf{x}_j$. This optimal solution is called the MWF, 
and the resulting estimator the linear minimum mean square error (LMMSE) estimator.
By using (\ref{eq:summodel}), (\ref{eq:srcstomodel}), and (\ref{eq:srccross}), the solution is rewritten as
\begin{align}
\mathbf{F}_{jn}^{MWF}=h_{jn}\mathbf{R}_{n}\mathbf{X}_j(\Theta)^{-1}.\label{eq:mwf2}
\end{align}
To design (\ref{eq:mwf2}), 
we need to estimate 
 $\Theta$ from the mixtures.

\subsection{Source Parameter Estimation}

The parameters $\Theta$ can be estimated by the maximum likelihood (ML) method, where the following negative log-likelihood is minimized:
\begin{align}
\mathcal{J}(\Theta)&\coloneqq -\ln p(\mathbf{x}_{1:J}\mid\Theta)\\
&=-\sum_{j=1}^J\ln p(\mathbf{x}_j\mid\Theta)\\
&\overset{c}{=}\sum_{j=1}^J[\ln\det\mathbf{X}_j(\Theta)+\mathbf{x}_j^H\mathbf{X}_j(\Theta)^{-1}\mathbf{x}_j],\label{eq:negllFCA}
\end{align}
where we have assumed sample independence and $\overset{c}{=}$ denotes equality up to a constant. This optimization can be performed by the expectation-maximization (EM) algorithm~\cite{Duong2010} or
the majorization-minimization (MM) algorithm~\cite{Sawada2012ICASSP,Yoshii2013,Sawada2013}. 

The EM algorithm is a generic optimization method applicable to the ML method. It
 uses auxiliary variables  called {\it latent variables} (denoted by $Z$) and a surrogate function called a {\it Q-function} (denoted by $Q$). The algorithm alternates the following two steps:
\begin{itemize}\itemsep 0pt
\item \textbf{E-step: }Update the posterior distribution  $p(Z\mid X,\Theta^\prime)$ of $Z$ given the data $X$ based on the parameter estimates $\Theta^\prime$ at the previous iteration. 
\item \textbf{M-step: }Update the parameters $\Theta$ so that the Q-function 
\begin{align}
Q(\Theta\mid\Theta^\prime)\coloneqq \int \ln p(X,Z\mid\Theta)p(Z\mid X,\Theta^\prime)dZ
\end{align}
does not decrease: $Q(\Theta\mid\Theta^\prime)\geq Q(\Theta^\prime\mid\Theta^\prime)$. 
\end{itemize}
The negative log-likelihood is guaranteed to be nonincreasing at each iteration. Moreover, there is no need for tuning a step size unlike in gradient methods.

In our case, $X=\mathbf{x}_{1:J}$ and $Z=\mathbf{c}_{1:J,1:N-1}$. Note that we exclude the $N$th source image from $Z$, which may seem a bit tricky.
The resulting algorithm consists in alternating the following update rules (see Appendix A for derivation):
\begin{align}
\bm{\Psi}_{jn}&\leftarrow \mathbf{F}_{jn}^{MWF}\widehat{\mathbf{X}}_j\bigl(\mathbf{F}_{jn}^{MWF}\bigr)^H
+\bigl(\mathbf{I}-\mathbf{F}_{jn}^{MWF}\bigr)(h_{jn}\mathbf{R}_n),\label{eq:psi}
\end{align}
\begin{align}
&h_{jn}\leftarrow \frac{1}{M}\tr(\mathbf{R}_{n}^{-1}\bm{\Psi}_{jn}),
&\mathbf{R}_{n}\leftarrow \frac{1}{J}\sum_{j=1}^J \frac{1}{h_{jn}}{\bm{\Psi}_{jn}}.\label{eq:FCAEM2}
\end{align}
Here,  $\mathbf{I}$ denotes the $M$th-order identity matrix, and 
\begin{align}
\widehat{\mathbf{X}}_j\coloneqq \mathbf{x}_j\mathbf{x}_j^H.\label{eq:Xhatdef}
\end{align}

Another useful optimization technique is applicable to FCA, namely a majorization-minimization (MM) algorithm. Let us consider minimizing a cost $\mathcal{J}(\Theta)$. The MM algorithm uses an auxiliary function $\mathcal{J}^+(\Theta,\Xi)$, which satisfies 
$\min_{\Xi}\mathcal{J}^+(\Theta,\Xi)=\mathcal{J}(\Theta)$.
That is, the minimum value of $\mathcal{J}^+$ with respect to $\Xi$ coincides with $\mathcal{J}$. Here, $\Xi$ is called auxiliary variables. The MM algorithm uses $\mathcal{J}^+$ to update $\Xi$ and $\Theta$ alternately. $\Xi$ is updated with $\Theta$ fixed by
$\Xi\leftarrow\arg\min_{\Xi}\mathcal{J}^+(\Theta,\Xi)$.
$\Theta$ is updated with $\Xi$ fixed so that $\mathcal{J}^+$  does not increase:
$\mathcal{J}^+(\Theta,\Xi)\leq \mathcal{J}^+(\Theta^\prime,\Xi)$ with $\Theta^\prime$ denoting the parameter estimates at the previous iteration.
Fig.~\ref{fig:MM} shows that $\mathcal{J}$ is guaranteed to be nonincreasing. Obviously, $\mathcal{J}^+$ is nonincreasing at each update of $\Xi$ or $\Theta$. Moreover, after each update of $\Xi$, the value of $\mathcal{J}^+$ coincides with that of $\mathcal{J}$. Therefore, $\mathcal{J}$ is guaranteed to be nonincreasing at each update of $\Xi$ or $\Theta$.
The EM algorithm can be regarded as a variant of the MM algorithm, where 
$\Xi$ consists of functions and $\mathcal{J}^+$ is a functional of $\Xi$.
\begin{figure}
\centering
\includegraphics[width=\columnwidth]{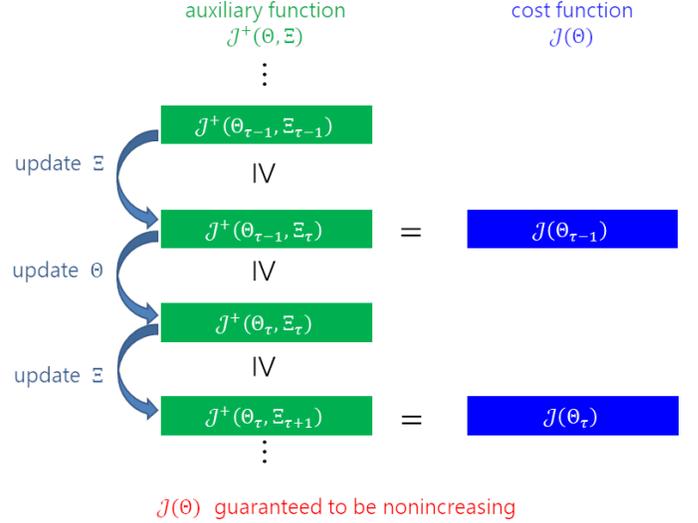}\\
\caption{MM algorithm. $\tau$ is the iteration index.}
\label{fig:MM}
\end{figure}

The resulting FCA algorithm consists in alternating the following update rules (see Appendix B for derivation):
\begin{align}
\mathbf{X}_j\leftarrow\sum_{n=1}^N h_{jn}\mathbf{R}_n,\label{eq:updateXMM}
\end{align}
\begin{align}
&h_{jn}\leftarrow
h_{jn}\sqrt{\frac{\tr\bigl(\mathbf{X}_j^{-1}\widehat{\mathbf{X}}_j\mathbf{X}_j^{-1}\mathbf{R}_n\bigr)}{\tr\bigl(\mathbf{X}_j^{-1}\mathbf{R}_n\bigr)}},
\end{align}
\begin{align}
\mathbf{R}_n\leftarrow
\Biggl(\sum_{j=1}^Jh_{jn}\mathbf{X}_j^{-1}\Biggr)^{-1}\#\Biggl[
\mathbf{R}_n\sum_{j=1}^J{h_{jn}}\mathbf{X}_j^{-1}\widehat{\mathbf{X}}_j\mathbf{X}_j^{-1}
\mathbf{R}_n
\Biggr].\label{eq:updateRMM}
\end{align}
Here, $\bm{\Omega}_1\#\bm{\Omega}_2\coloneqq \bm{\Omega}_1^\frac{1}{2}(\bm{\Omega}_1^{-\frac{1}{2}}\bm{\Omega}_2\bm{\Omega}_1^{-\frac{1}{2}})^\frac{1}{2}\bm{\Omega}_1^\frac{1}{2}\in S_{++}^M$ denotes the geometric mean\footnote{$\bm{\Omega}_1\#\bm{\Omega}_2$ is the unique positive definite solution $\mathbf{X}$ of an algebraic Riccati equation $\mathbf{X}\bm{\Omega}_1^{-1}\mathbf{X}=\bm{\Omega}_2$.
Geometrically,
$\bm{\Omega}_1\#\bm{\Omega}_2$ can be interpreted as the midpoint of the geodesic connecting $\bm{\Omega}_1$ and $\bm{\Omega}_2$.} of $\bm{\Omega}_1,\bm{\Omega}_2\in S_{++}^M$~\cite{Pedersen1972,Pusz1975,Kubo1980}.
For $\bm{\Omega}\in S_{++}^M$ and $r\in\mathbb{R}$, $\bm{\Omega}^r\coloneqq \mathbf{U}\diag(\sigma_1^r,\ldots,\sigma_M^r)\mathbf{U}^H$, where $\mathbf{U}\diag(\sigma_1,\ldots,\sigma_M)\mathbf{U}^H$ ($\mathbf{U}$: unitary) is an eigenvalue decomposition of  $\bm{\Omega}$ and $\diag(\sigma_1,\ldots,\sigma_M)$ denotes the diagonal matrix composed of $\sigma_1,\ldots,\sigma_M$.
$\bm{\Omega}_1\#\bm{\Omega}_2$ is a matrix extension of the geometric mean $\omega_1\#\omega_2=(\omega_1\omega_2)^\frac{1}{2}$ of two positive numbers $\omega_1,\omega_2$.
Indeed, it is straightforward to check that the former reduces to the latter, when $M=1$.

\subsection{Drawback}
\label{sec:drawback}
Although applicable even to the underdetermined case, FCA has the drawback of being computationally expensive because of massive numbers of matrix inversions required in the EM/MM algorithm. Indeed, matrix inversion is required at each time-frequency point and each EM/MM iteration. We can see this by looking at (\ref{eq:psi}), where $\mathbf{X}_j(\Theta)$ in $\mathbf{F}_{jn}^{MWF}=h_{jn}\mathbf{R}_{n}\mathbf{X}_j(\Theta)^{-1}$ depends on both the time and the frequency (recall that the frequency bin index $i$ is being omitted). Since the number of time-frequency points is normally huge, this results in numerous matrix inversions. For example, consider a typical case of 8-s mixtures sampled at 16\,kHz and STFT with a frame length of 64\,ms and a shift of 32\,ms. This results in as much as 127,737 time-frequency points.

\section{FastFCA: Computationally Efficient Extension of FCA}
\label{sec:prop}


This section describes
FastFCA, a computationally efficient extension of FCA.

\subsection{Motivation}
\label{sec:approach}

To overcome the above drawback of numerous matrix inversions in FCA, 
we exploit
the well-known property of diagonal matrices that matrix inversion amounts to mere inversion of diagonal elements,
which can be performed efficiently.
This implies that, if the $N$ SCMs $\mathbf{R}_{1:N}$ of all $N$ sources were diagonal, the matrix inversions in FCA could be performed efficiently. Unfortunately, however, they are far from being diagonal in practice, because of significant signal correlation between microphones.

This motivates us to consider joint diagonalization (by congruence) of $\mathbf{R}_{1:N}$:
\begin{align}
\begin{cases}
\mathbf{W}^H\mathbf{R}_{1}\mathbf{W}=\bm{\Lambda}_{1},\\
\ \vdots\\
\mathbf{W}^H\mathbf{R}_{N}\mathbf{W}=\bm{\Lambda}_{N}.
\end{cases}\label{eq:jd}
\end{align}
Here, $\mathbf{W}\in GL(M)$ denotes a nonsingular matrix (unknown), and 
$\bm{\Lambda}_1,\ldots,\bm{\Lambda}_N\in S_{++}^{M}$ diagonal matrices (unknown), where $GL(M)$ denotes the set of $M$th-order complex nonsingular matrices.
When there are only two sources ({\it i.e.}, $N=2$), we can obtain such $\mathbf{W}$ and $\bm{\Lambda}_{1:2}$
 by 
solving a generalized eigenvalue problem (see~\cite{Horn1990}).
We exploited this property to make FCA computationally efficient when $N=2$~\cite{Ito2018EUSIPCOFastFCA}, which is an earlier idea of FastFCA.
However, when there are more than two sources ({\it i.e.}, $N\geq 3$), such $\mathbf{W}$ and $\bm{\Lambda}_{1:N}$ may not exist, depending on the values of $\mathbf{R}_{1:N}$.

\subsection{Approach: Joint Diagonalizability (JD) Constraint}

To make FCA computationally efficient for an arbitrary number of sources, FastFCA restricts $\mathbf{R}_{1:N}$ to jointly diagonalizable matrices. 
In other words, it imposes on $\mathbf{R}_{1:N}$ the constraint that there exist a nonsingular matrix $\mathbf{W}\in GL(M)$ and diagonal matrices $\bm{\Lambda}_{1},\ldots,\bm{\Lambda}_N\in S_{++}^M$ satisfying (\ref{eq:jd}), which we call a {\it joint diagonalizability (JD) constraint}.
In this case, $\mathbf{R}_{1:N}$ can be parametrized by the matrices $\mathbf{W}$ and 
$\bm{\Lambda}_{1:N}$
as 
\begin{align}
\mathbf{R}_{n}=\mathbf{W}^{-H}\bm{\Lambda}_{n}\mathbf{W}^{-1},\label{eq:parametrizationFastFCA}
\end{align}
where 
$\mathbf{W}^{-H}\coloneqq (\mathbf{W}^H)^{-1}=(\mathbf{W}^{-1})^H$.
Here, $\mathbf{W}$ and $\bm{\Lambda}_{1:N}$ are considered as model parameters to be estimated.
Note that $\mathbf{W}$ is not restricted to unitary matrices unlike in NMF with transform learning~\cite{Fagot2018}. 

\subsection{Signal Model in FastFCA}
\label{sec:model}
In FastFCA, the mixtures are modeled by (\ref{eq:summodel}), and each source image by (\ref{eq:srcstomodel}) as in FCA.
The difference from FCA lies in 
that  $\mathbf{R}_n$ is not a free parameter here but
parametrized as in (\ref{eq:parametrizationFastFCA}).
Consequently,  (\ref{eq:srcstomodel}) becomes
\begin{align}
p(\mathbf{c}_{jn}\mid\Theta)
=\mathcal{N}_c
(\mathbf{c}_{jn}\mid\mathbf{0},h_{jn}\mathbf{W}^{-H}\bm{\Lambda}_{n}\mathbf{W}^{-1}),\label{eq:FastFCAstosigmodel}
\end{align}
where $\Theta$ is now defined as $\Theta\coloneqq \{\mathbf{W},\mathbf{\Lambda}_{1:N},h_{1:J,1:N}\}$.
Thus, $\mathbf{W}^{H}\mathbf{c}_{jn}$ has a distribution 
\begin{align}
p(\mathbf{W}^{H}\mathbf{c}_{jn}\mid\Theta)
=\mathcal{N}_c
(\mathbf{W}^{H}\mathbf{c}_{jn}\mid\mathbf{0},h_{jn}\bm{\Lambda}_{n}).\label{eq:covimgfast}
\end{align}
Since $h_{jn}\bm{\Lambda}_{n}$ is diagonal, (\ref{eq:covimgfast}) implies that the elements of $\mathbf{W}^{H}\mathbf{c}_{jn}$ are mutually uncorrelated and thus mutually independent as well because of Gaussianity.
Define
\begin{align}
\mathbf{y}_j\coloneqq \mathbf{W}^H\mathbf{x}_j=\sum_{n=1}^N\mathbf{W}^H\mathbf{c}_{jn}.\label{eq:decormix}
\end{align}
From (\ref{eq:covimgfast}) and (\ref{eq:decormix}), $\mathbf{y}_j$ has a distribution 
\begin{align}
p(\mathbf{y}_{j}\mid\Theta)
=\mathcal{N}_c
(\mathbf{y}_{j}\mid\mathbf{0},\mathbf{Y}_j),\label{eq:FastFCAstosigmodel2}
\end{align}
with $\mathbf{Y}_j$ being a diagonal covariance matrix 
\begin{align}
\mathbf{Y}_j\coloneqq\sum_{n=1}^Nh_{jn}\bm{\Lambda}_{n}.\label{eq:Ydef}
\end{align}
Again, this implies that the elements of $\mathbf{y}_j$ are mutually uncorrelated as well as mutually independent. In this sense,
we call $\mathbf{W}^H$ a {\it decorrelation matrix}, and  
$\mathbf{y}_j$ {\it decorrelated mixtures}.


\subsection{Cost}
In the following, we summarize our assumptions in FastFCA:
\begin{enumerate}
\item \textbf{Linear transform:} As in (\ref{eq:decormix}), the decorrelated mixtures $\mathbf{y}_j$ equal the mixtures $\mathbf{x}_j$ transformed by the decorrelation matrix $\mathbf{W}^H$. Thus, their distributions are related by 
\begin{align}
p(\mathbf{x}_j\mid\Theta)=|\det\mathbf{W}|^2p(\mathbf{y}_j\mid\Theta).\label{eq:as1}
\end{align}
\item \textbf{Output independence:} As in (\ref{eq:FastFCAstosigmodel2}), the decorrelated mixtures $\mathbf{y}_j=(y_{1j},\ldots,y_{Mj})^T$ ($^T$: transposition) are independent:
\begin{align}
p(\mathbf{y}_j\mid\Theta)=\prod_{m=1}^Mp(y_{mj}\mid\Theta).
\end{align}
Here, the marginal distribution of each decorrelated mixture  is given by 
\begin{align}
p(y_{mj}\mid\Theta)=\mathcal{N}_{\mathbb{C}}(y_{mj}\mid 0,\sigma_{mj}^2)\label{eq:margindecor}
\end{align}
with 
\begin{align}
\sigma_{mj}^2\coloneqq [\mathbf{Y}_j]_{mm}=\sum_{n=1}^Nh_{jn}[\bm{\Lambda}_n]_{mm}.\label{eq:sigmadef}
\end{align}
Here,
 $[\bm{\Lambda}_n]_{mm}$ denotes the $(m,m)$-element of $\bm{\Lambda}_n$. 
\item \textbf{Sample independence:} The samples $\mathbf{x}_{1:J}$ are mutually independent, so that
\begin{align}
p(\mathbf{x}_{1:J}\mid\Theta)=\prod_{j=1}^Jp(\mathbf{x}_{j}\mid\Theta).\label{eq:as3}
\end{align}
\end{enumerate}

By putting the above assumptions together, we obtain the negative log-likelihood to be minimized as follows:
\begin{align}
&-\ln p(\mathbf{x}_{1:J}\mid\Theta)\notag\\
&=-J\ln|\det\mathbf{W}|^2
-\sum_{m=1}^M\sum_{j=1}^J\ln p(y_{mj}\mid\Theta)\\
&\overset{c}{=}-J\ln|\det\mathbf{W}|^2+\sum_{m=1}^M\sum_{j=1}^J\frac{|y_{mj}|^2}{\sigma^2_{mj}}
+\sum_{m=1}^M\sum_{j=1}^J\ln\sigma^2_{mj}.\label{eq:negll}
\end{align}

\subsection{Parameter Estimation}
Interestingly,  the FastFCA cost (\ref{eq:negll}) can be viewed as a `mixture' of ICA and NMF costs, in a similar way as in~\cite{Kitamura2016}.
On the one hand, the first two terms in  (\ref{eq:negll}), relevant to updating $\mathbf{W}$, turn out to be the cost of time-varying Gaussian ICA\footnote{To be precise, here we refer to a
complex extension of time-varying Gaussian ICA~\cite{Pham2001}, which originally operates on real signals.}~\cite{Pham2001}, up to a constant independent of $\mathbf{W}$.
Note here that $y_{mj}$ depends on $\mathbf{W}$ as $y_{mj}=\mathbf{w}_m^H\mathbf{x}_j$ with $\mathbf{W}=(\mathbf{w}_1,\ldots,\mathbf{w}_M)$.
On the other hand, the last two terms, relevant to updating $\bm{\Lambda}_{n}$ and $h_{jn}$, turn out to be
the cost of NMF using the Itakura-Saito divergence (IS-NMF)~\cite{Fevotte2009}, up to a constant independent of $\bm{\Lambda}_{n}$ and $h_{jn}$.
Consequently, we can minimize  (\ref{eq:negll}) efficiently
by leveraging efficient optimization algorithms 
that have been developed for ICA and NMF\footnote{The following are some further implications of
this interpretation of the FastFCA cost.
(\ref{eq:negll})
 reduces to the cost of
time-varying Gaussian ICA, when $N=M$ and $\bm{\Lambda}_n$ is fixed to $\mathbf{e}_n\mathbf{e}_n^T$.
Moreover, (\ref{eq:negll})
 reduces to the 
IS-NMF cost, when $\mathbf{W}$ is fixed to the identity matrix $\mathbf{I}$.
As such, FastFCA can be viewed as an extension of both time-varying Gaussian ICA and IS-NMF.
}.

Consider updating
$\mathbf{W}$ with $\mathbf{\Lambda}_{1:N}$ and $h_{1:J,1:N}$ fixed based on ICA.
For example, we can use ICA algorithms, such as natural gradient~\cite{Amari1995} and  
iterative projection (IP)~\cite{Ono2011WASPAA,Ono2012APSIPA}. Here we focus on the latter,
which has the advantages that (\ref{eq:negll}) is guaranteed to be
nonincreasing and that
 it is free of hyperparameters, such as the step size, often requiring fine tuning.
IP is based on block coordinate descent, where one column of $\mathbf{W}$ is updated at a time, instead of the whole matrix $\mathbf{W}$. The stationary condition for the $m$th column
$\mathbf{w}_m$ is given by
\begin{align}
\mathbf{W}^H\mathbf{Q}_m\mathbf{w}_m=\mathbf{e}_m,\label{eq:stationary}
\end{align}
where $\mathbf{Q}_m$ is a weighted covariance matrix given by
\begin{align}
\mathbf{Q}_m\coloneqq \frac{1}{J}\sum_{j=1}^J\frac{1}{\sigma_{mj}^2}\widehat{\mathbf{X}}_j\label{eq:wsampcov}
\end{align}
and  $\mathbf{e}_m$ denotes the $m$th column of $\mathbf{I}=(\mathbf{e}_1,\ldots,\mathbf{e}_M)$. 
A solution to
(\ref{eq:stationary}) can be obtained by the following two-step procedure\footnote{We originally used a slight variant of this procedure for FastFCA in~\cite{Ito2018IWAENC} and for its full-band extension called FastMNMF (see Section~\ref{sec:fastmnmfmain}) in~\cite{Ito2019ICASSP}, where 
$\mathbf{w}_m$ was normalized by
$\mathbf{w}_m\leftarrow \mathbf{w}_m/[\mathbf{w}_m]_m$ ($[\mathbf{w}_m]_m$: the $m$th element of $\mathbf{w}_m$) instead of (\ref{eq:ip2}) (see \cite{Ito2019ICASSP}).
However, it was unclear whether 
the negative log-likelihood is guaranteed to be nonincreasing in this case.
Consequently, 
 (\ref{eq:ip2}) was used for FastMNMF in~\cite{Sekiguchi2020}, where the negative log-likelihood is guaranteed to be nonincreasing.}:
\begin{align}
\mathbf{w}_m&\leftarrow(\mathbf{W}^{H}\mathbf{Q}_{m})^{-1}\mathbf{e}_m,\label{eq:ip1}\\
\mathbf{w}_m&\leftarrow\frac{\mathbf{w}_m}{\sqrt{\mathbf{w}_m^H\mathbf{Q}_{m}\mathbf{w}_m}}.\label{eq:ip2}\end{align}
The matrix $\mathbf{W}$ can be updated by updating $\mathbf{w}_m$ by the above procedure
for all $m=1,\ldots,M$.
 
Now consider updating
$\mathbf{\Lambda}_{1:N}$ and $h_{1:J,1:N}$ with $\mathbf{W}$ fixed based on NMF.
In this case, the last two terms of (\ref{eq:negll}) are relevant, which are rewritten as
\begin{align}
\sum_{m=1}^M\sum_{j=1}^J D_{IS}(|y_{mj}|^2\mid\sigma^2_{mj})\label{eq:neglllast2}
\end{align}
up to a constant independent of $\mathbf{\Lambda}_{1:N}$ and $h_{1:J,1:N}$.
Here,
$D_{IS}(\omega_1\mid\omega_2)\coloneqq \omega_1\omega_2^{-1}-\ln(\omega_1\omega_2^{-1})-1$ denotes the
Itakura-Saito divergence of $\omega_1,\omega_2>0$.
Define a non-negative matrix $\mathbf{U}$ by
\begin{align}
[\mathbf{U}]_{mj}\coloneqq |y_{mj}|^2=\mathbf{w}_m^H\widehat{\mathbf{X}}_j\mathbf{w}_m,\label{eq:Udef} 
\end{align}
and non-negative matrices
$\mathbf{L}\in\mathbb{R}_{++}^{M\times N}$ and $\mathbf{H}\in\mathbb{R}_{++}^{N\times J}$ by
\begin{align}
&[\mathbf{L}]_{mn}\coloneqq [\bm{\Lambda}_n]_{mm},
&[\mathbf{H}]_{nj}\coloneqq h_{jn}\label{eq:defLH}
\end{align}
with $\mathbb{R}_{++}$ being the set of positive numbers. Since (\ref{eq:sigmadef}) becomes
$\sigma_{mj}^2=[\mathbf{LH}]_{mj}$,
(\ref{eq:neglllast2}) can be rewritten as $\sum_{m=1}^M\sum_{j=1}^J D_{IS}([\mathbf{U}]_{mj}\mid[\mathbf{LH}]_{mj})$,
which is nothing but the cost of NMF\footnote{Note that this is NMF in the channel-time~\cite{Togami2010,Chiba2014} rather than the time-frequency domain~\cite{Smaragdis2003,Virtanen2007}.
} $\mathbf{U}\approx\mathbf{LH}$ using the Itakura-Saito divergence (IS-NMF)~\cite{Fevotte2009}.
Consequently, we can update 
$\mathbf{\Lambda}_{1:N}$ and $h_{1:J,1:N}$, or equivalently $\mathbf{L}$ and $\mathbf{H}$, by using
 an EM~\cite{Fevotte2009}
or an MM algorithm~\cite{Nakano2010} for NMF, where (\ref{eq:negll}) is guaranteed to be
nonincreasing.

Therefore, our optimization scheme for minimizing (\ref{eq:negll}) consists in alternating the following two steps:
\begin{itemize}\itemsep 0pt
\itemsep 0pt
\item Update $\mathbf{W}$ with $\mathbf{L}$ and $\mathbf{H}$ fixed based on IP.
\item Update $\mathbf{L}$ and $\mathbf{H}$ with $\mathbf{W}$ fixed based on an EM or an MM algorithm for NMF.
\end{itemize}
Algorithm~\ref{algo:FastFCA} shows the pseudocode of an {\it IP+EM} algorithm,
which uses
 the EM algorithm to update $\mathbf{L}$ and $\mathbf{H}$.
Algorithm~\ref{algo:FastFCAMM} shows the pseudocode of an {\it IP+MM} algorithm,
which uses 
the MM algorithm
 to update $\mathbf{L}$ and $\mathbf{H}$.
$\odot$ denotes element-wise multiplication, 
$\circ$ element-wise exponentiation,
$\bm{1}_{M\times J}$ the matrix of all ones of size $M\times J$,
$\displaystyle[\mathbf{P}]_{:,n}$ the $n$th column of a matrix $\mathbf{P}$,
and $\displaystyle[\mathbf{P}]_{n,:}$ the $n$th row of a matrix $\mathbf{P}$.

\begin{algorithm}
\caption{IP+EM algorithm for FastFCA}
\label{algo:FastFCA}
\hspace*{\algorithmicindent}\textbf{Input:} $\widehat{\mathbf{X}}_{1:J}, N$\\
\hspace*{\algorithmicindent}\textbf{Output:} $\mathbf{W}$, $\mathbf{L}$, $\mathbf{H}$
\begin{algorithmic}[1]
\State Initialize  $\mathbf{W}$, $\mathbf{L}$, and $\mathbf{H}$
\Repeat
\For{$m=1:M$}
\State $\mathbf{Q}_{m}\leftarrow\frac{1}{J}\sum_{j=1}^J\frac{1}{[\mathbf{LH}]_{mj}}\widehat{\mathbf{X}}_j$
\State Update $\mathbf{w}_m$ by (\ref{eq:ip1}) and (\ref{eq:ip2})
\EndFor
\State Update $\mathbf{U}$ by $[\mathbf{U}]_{mj}\leftarrow\mathbf{w}_m^H\widehat{\mathbf{X}}_{j}\mathbf{w}_m$
\For{$n=1:N$}\vspace{1mm}
\State $\displaystyle\mathbf{G}_n\leftarrow ([\mathbf{L}]_{:,n}[\mathbf{H}]_{n,:})\odot(\mathbf{L}\mathbf{H})^{\circ -1}$\vspace{1mm}
\State $\displaystyle {\bm{\Phi}}_{n}\leftarrow 
(\mathbf{G}_n)^{\circ 2}\odot \mathbf{U}+(\mathbf{1}_{M\times J}-\mathbf{G}_n)\odot
([\mathbf{L}]_{:,n}[\mathbf{H}]_{n,:})$\vspace{1mm}
\State $\displaystyle[\mathbf{L}]_{:,n}\leftarrow\frac{1}{J}{\bm{\Phi}}_{n}\Bigl[(\mathbf{H}^T)^{\circ -1}\Bigr]_{:,n}$\vspace{1mm}
\State $\displaystyle[\mathbf{H}]_{n,:}\leftarrow\frac{1}{M}\Bigl[(\mathbf{L}^{T})^{\circ -1}\Bigr]_{n,:}{\bm{\Phi}}_{n}$\vspace{1mm}
\EndFor
\Until{some stop condition is met}
\end{algorithmic}
\end{algorithm}

\begin{algorithm}
\caption{IP+MM algorithm for FastFCA}
\label{algo:FastFCAMM}
\hspace*{\algorithmicindent}\textbf{Input:} $\widehat{\mathbf{X}}_{1:J}, N$\\
\hspace*{\algorithmicindent}\textbf{Output:} $\mathbf{W}$, $\mathbf{L}$, $\mathbf{H}$
\begin{algorithmic}[1]
\State Initialize  $\mathbf{W}$, $\mathbf{L}$, and $\mathbf{H}$
\Repeat
\For{$m=1:M$}
\State $\mathbf{Q}_{m}\leftarrow\frac{1}{J}\sum_{j=1}^J\frac{1}{[\mathbf{LH}]_{mj}}\widehat{\mathbf{X}}_j$
\State Update $\mathbf{w}_m$ by (\ref{eq:ip1}) and (\ref{eq:ip2})
\EndFor
\State Update $\mathbf{U}$ by $[\mathbf{U}]_{mj}\leftarrow\mathbf{w}_m^H\widehat{\mathbf{X}}_{j}\mathbf{w}_m$
\State $\displaystyle\mathbf{L}\leftarrow\mathbf{L}\odot\Biggl\{{\frac{[{\mathbf{U}}\odot(\mathbf{L}\mathbf{H})^{\circ -2}]\mathbf{H}^T}{(\mathbf{L}\mathbf{H})^{\circ -1}\mathbf{H}^T}}\Biggr\}^{\circ\frac{1}{2}}$\vspace{1mm}
\State $\displaystyle\mathbf{H}\leftarrow\mathbf{H}\odot\Biggl\{{\frac{\mathbf{L}^T[{\mathbf{U}}\odot(\mathbf{L}\mathbf{H})^{\circ -2}]}{\mathbf{L}^T(\mathbf{L}\mathbf{H})^{\circ -1}}}\Biggr\}^{\circ\frac{1}{2}}$
\Until{some stop condition is met}
\end{algorithmic}
\end{algorithm}

\subsection{Multichannel Wiener Filter}
\label{sec:MWFFastFCA}

Once the parameters have been estimated, the source images can be estimated by the MWF. Here we show that the filter has an interesting form in the case of FastFCA. By plugging (\ref{eq:parametrizationFastFCA}) in (\ref{eq:mwf2}), we have
\begin{align}
&\mathbf{F}_{jn}^{MWF}\notag\\
&=\mathbf{W}^{-H}\begin{pmatrix}
\frac{h_{jn}[\bm{\Lambda}_{n}]_{11}}{\sum_{\nu=1}^N h_{j\nu}[\bm{\Lambda}_{\nu}]_{11}}&&\mathbf{O}\\
&\ddots&\\
\mathbf{O}&&\frac{h_{jn}[\bm{\Lambda}_{n}]_{MM}}{\sum_{\nu=1}^N h_{j\nu}[\bm{\Lambda}_{\nu}]_{MM}}\end{pmatrix}\mathbf{W}^{H}.\label{eq:Doclo}
\end{align}
Therefore, as in Fig.~\ref{fig:decompmwf}, the MWF  decomposes into three parts: 1) decorrelation, 2) single-channel Wiener filters in the decorrelation domain, and 3)
 the inverse transform (projection back). Note that (\ref{eq:Doclo}) is a generalization of a result in \cite{Doclo2002}.
(They consider the case $N=2, h_{jn}=1$ in  the context of denoising.)
%
%

\begin{figure}
\centering
\includegraphics[width=\columnwidth]{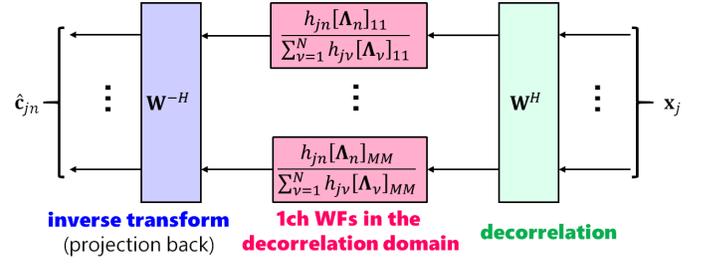}\\
\caption{Decomposition of the MWF for FastFCA.}
\label{fig:decompmwf}
\end{figure}

\subsection{FastMNMF: Full-band Extension of FastFCA}
\label{sec:fastmnmfmain}
We can extend FastFCA to a full-band BSS method by incorporating NMF as a source signal model.
The resulting method is called {\it FastMNMF}~\cite{Ito2019ICASSP,Sekiguchi2020}\footnote{We first published the 
FastMNMF model 
 in a conference paper~\cite{Ito2019ICASSP}, and so did Sekiguchi {\it et al.}~\cite{Sekiguchi2019} slightly later and independently.}, because it can also be considered as a computationally
efficient extension of full-rank MNMF~\cite{Arberet2010,Sawada2012ICASSP,Sawada2013}.
While FastFCA processes each frequency component independently, 
FastMNMF processes all frequency components jointly, whereby
enabling permutation-free BSS. See Appendix~\ref{sec:fastmnmf} for more details.

\section{A Unified Framework for Underdetermined and Determined Audio BSS}
\label{sec:relICA}

\subsection{Description of Unified Framework}

This section presents a unified framework for underdetermined and determined audio BSS, which
highlights a close theoretical connection between FastFCA and other methods.
In this section, we make the frequency bin index $i$ explicit.

As shown in Fig.~\ref{fig:param}, various methods share 
the negative log-likelihood of the same functional form.
As shown in Fig.~\ref{fig:framework}, the difference between these methods lies in whether or not three constraints (represented by the three axes) are imposed on model parameters: 
\begin{itemize}\itemsep 0pt
\item the joint diagonalizability (JD) constraint on $\mathbf{R}_{in}$, 
\item a rank-1 constraint on $\mathbf{R}_{in}$, 
\item an NMF constraint on $h_{ijn}$.
\end{itemize}
In FCA, none of these constraints is imposed on $\mathbf{R}_{in}$ and $h_{ijn}$.
In R1CA, the rank-1 constraint is imposed on
$\mathbf{R}_{in}$.
 In FastFCA, the JD constraint is imposed on $\mathbf{R}_{in}$.
As shown in Section~\ref{sec:tvgderiv},
in time-varying Gaussian ICA, 
both the JD and the rank-1 constraints are imposed on $\mathbf{R}_{in}$,
 leading to the parametrization 
\begin{align}
\mathbf{R}_{in}=\mathbf{W}_i^{-H}\mathbf{e}_n\mathbf{e}_n^T\mathbf{W}_i^{-1}\label{eq:ICAparam}
\end{align}
in Fig.~\ref{fig:param}. 
The restriction to the determined case is also shown to be an immediate consequence of the simultaneous presence of both constraints (see Section~\ref{sec:tvgderiv}).
(\ref{eq:ICAparam}) will be derived in the next paragraph.
The other methods are obtained by imposing the NMF constraint on $h_{ijn}$.

Now we derive the parametrization (\ref{eq:ICAparam}) in time-varying Gaussian ICA.
This method is applicable to the determined case $N=M$,
where the signal model is given by $\mathbf{x}_{ij}=\mathbf{A}_i\mathbf{y}_{ij}=\mathbf{W}^{-H}_i\mathbf{y}_{ij}$.
Here, 
$\mathbf{A}_i\in GL(M)$ denotes
the mixing matrix,
$\mathbf{W}^H_i\coloneqq \mathbf{A}^{-1}_i$ the separation matrix,
and $\mathbf{y}_{ij}\coloneqq
(y_{ij1},\ldots,y_{ijN})^T\in\mathbb{C}^N$ the source signals.
We assume that the source signals independently follow 
zero-mean proper complex Gaussian distributions with time-varying variances $h_{ijn}$ as
$
p(\mathbf{y}_{ij}\mid\Theta)=\prod_{n=1}^N\mathcal{N}_c(y_{ijn}\mid {0},h_{ijn})=\mathcal{N}_c(\mathbf{y}_{ij}\mid \mathbf{0},\sum_{n=1}^Nh_{ijn}\mathbf{e}_n\mathbf{e}_n^T)$
with $\Theta\coloneqq \{\mathbf{W}_{1:I},h_{1:I,1:J,1:N}
\}$.
The distribution of the mixtures is therefore given by
$
p(\mathbf{x}_{ij}\mid\Theta)
=\mathcal{N}_c(\mathbf{x}_{ij}\mid \mathbf{0},\sum_{n=1}^Nh_{ijn}\mathbf{R}_{in})$ with $\mathbf{R}_{in}$ given by (\ref{eq:ICAparam}).

%
\begin{figure}
\centering
\includegraphics[width=\columnwidth]{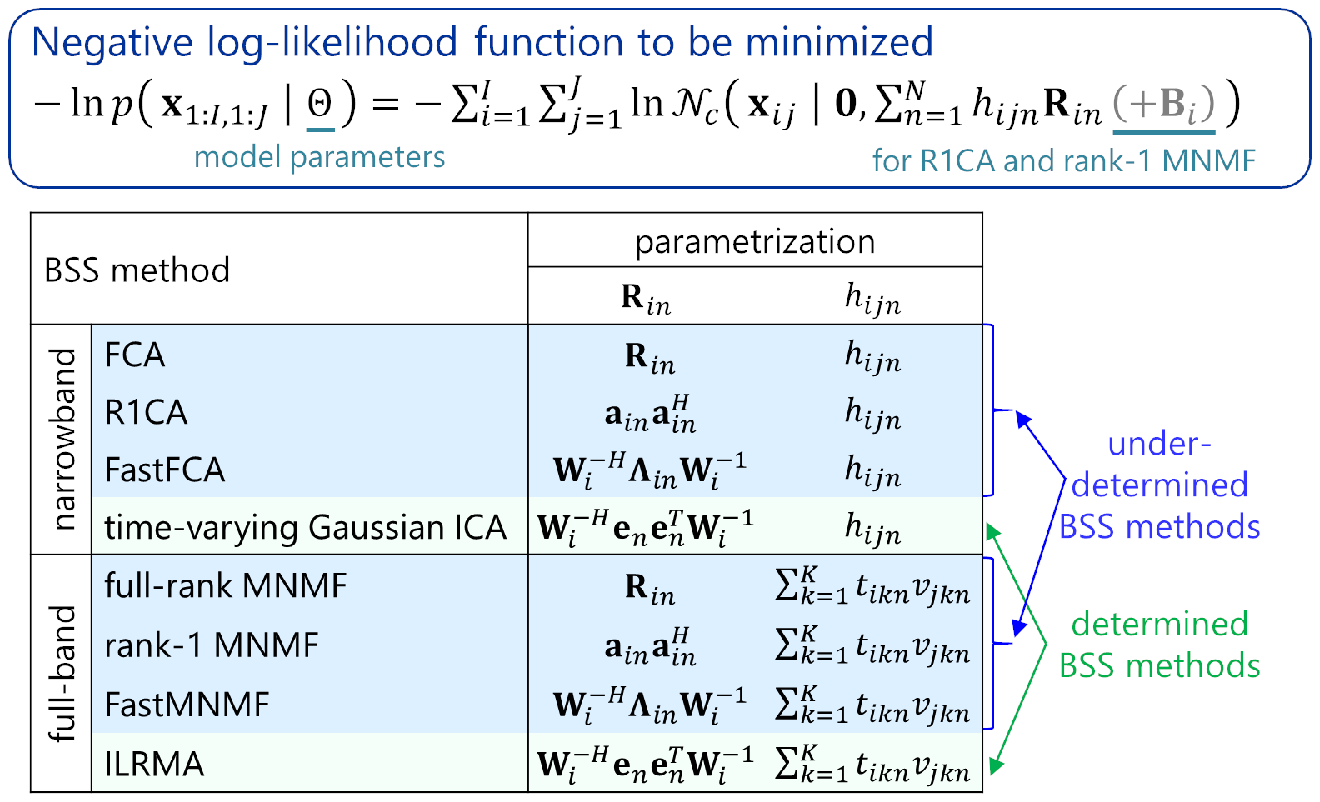}\\
\caption{Unified framework: parametrization. $\mathbf{B}_{i}$ denotes the noise covariance matrix for R1CA and rank-1 MNMF.}
\label{fig:param}
\end{figure}

\subsection{Properties of BSS Methods}
\begin{figure}
\centering
\includegraphics[width=\columnwidth]{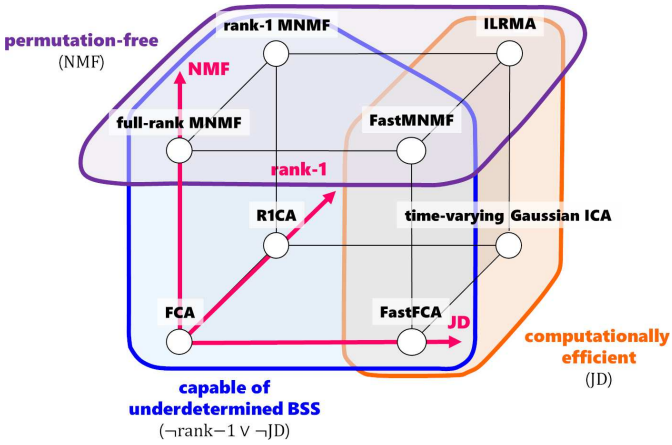}\\
\caption{Various properties of BSS methods.}
\label{fig:properties}
\end{figure}

As shown in Fig.~\ref{fig:properties},
various properties of BSS methods can be explained by using the framework in Fig.~\ref{fig:framework}. As we have just described, when both the JD and the rank-1 constraints are imposed on $\mathbf{R}_{in}$, the determined condition is automatically derived. This implies that underdetermined BSS is possible only when either the rank-1 or the JD constraint is absent as shown in Fig.~\ref{fig:properties}.
Moreover, computationally efficient BSS is possible when the JD constraint is present.
Furthermore, permutation-free BSS is possible when the NMF constraint is present.

These properties highlight the advantage of FastFCA and its full-band extension, namely FastMNMF, over the other methods. Time-varying Gaussian ICA is computationally efficient because of the JD constraint but incapable of underdetermined BSS because of the simultaneous presence of the JD and the rank-1 constraints. On the other hand, R1CA and FCA are capable of underdetermined BSS by removing the JD constraint, but this is at the cost of computational efficiency. In contrast, FastFCA is not only capable of underdetermined BSS but also computationally efficient by removing the rank-1 constraint but retaining the JD constraint. Similar statements hold for the full-band BSS methods.

Here, an interesting question arises: Can 
 R1CA (respectively rank-1 MNMF) be made faster based on the JD constraint, just as FCA (respectively full-rank MNMF)? The answer is yes.
The resulting method is nothing but time-varying Gaussian ICA (respectively ILRMA). We can consider
time-varying Gaussian ICA to be `FastR1CA'
(respectively ILRMA to be `Fast rank-1 MNMF').
However, time-varying Gaussian ICA (respectively ILRMA) is restricted to the determined case unlike FastFCA (respectively FastMNMF).

\subsection{Mathematical Background}
\label{sec:tvgderiv}
In the following, we show that the following statements hold under certain mild conditions:
\begin{itemize}\itemsep 0pt
\itemsep 0pt
\item Imposing
 both the JD and the rank-1 constraints on $\mathbf{R}_{in}$ leads to the parametrization (\ref{eq:ICAparam}) in time-varying Gaussian ICA.
\item The restriction to the determined case in time-varying Gaussian ICA is an immediate linear algebraic consequence of the simultaneous presence of both constraints. 
\item Under (\ref{eq:ICAparam}), the MWF reduces to the separation matrix followed by projection back, commonly used in ICA.
\end{itemize}

First of all, we have Theorem~\ref{thm:shrink} below:
\begin{definition}
Let $M,N$ be positive integers.
$\mathbf{R}_1,\ldots,\mathbf{R}_N\in
\mathbb{C}^{M\times M}$ are said to be jointly diagonalizable, if $\mathbf{W}^H\mathbf{R}_n\mathbf{W}$ is diagonal for all $n\in\{1,\ldots,N\}$ for some $\mathbf{W}\in GL(M)$.
\label{def:jd}
\end{definition}
\begin{theorem}
Let $M,N$ be positive integers satisfying $M<N$, and $\mathbf{a}_1,\ldots,\mathbf{a}_N\in\mathbb{C}^M-\{\mathbf{0}\}.$
If $\mathbf{a}_1\mathbf{a}_1^H,\ldots,\mathbf{a}_N\mathbf{a}_N^H$ are jointly diagonalizable,
$\mathbf{a}_n\propto\mathbf{a}_{n^\prime}$ for some
distinct $n$ and $n^\prime$.
\label{thm:shrink}
\end{theorem}
\begin{proof}
Since $\mathbf{a}_1\mathbf{a}_1^H,\ldots,\mathbf{a}_N\mathbf{a}_N^H$ are jointly diagonalizable,
there exist $\mathbf{W}\in GL(M)$ such that, for all $n\in\{1,\ldots,N\}$,
 $\mathbf{W}^H\mathbf{a}_n\mathbf{a}_n^H\mathbf{W}$ is diagonal.
Moreover, since  $\mathbf{W}^H\mathbf{a}_n\mathbf{a}_n^H\mathbf{W}$ is Hermitian positive semidefinite and has rank 1,
there exist 
$\lambda_1,\ldots,\lambda_N>0$
and $\varphi: \{1,\ldots,N\}\rightarrow \{1,\ldots,M\}$
such that, for all $n\in\{1,\ldots,N\}$,
$
\mathbf{W}^H\mathbf{a}_n\mathbf{a}_n^H\mathbf{W}=\lambda_n\mathbf{e}_{\varphi(n)}\mathbf{e}_{\varphi(n)}^T.
$
Thus, there exist 
$\theta_1,\ldots,\theta_N\in\mathbb{R}$
such that, for all $n\in\{1,\ldots,N\}$, $\mathbf{W}^H\mathbf{a}_n=\sqrt{\lambda_n}\exp(\sqrt{-1}\theta_n)\mathbf{e}_{\varphi(n)}$.
Since $M<N$, $\varphi$ is not injective. Therefore,
there exist  distinct 
$n$ and $n^\prime$ such that
$\mathbf{a}_n=\sqrt{\lambda_{n}/\lambda_{n^\prime}}\exp\bigl[\sqrt{-1}(\theta_{n}-\theta_{n^\prime})\bigr]\mathbf{a}_{n^\prime}$.
\end{proof}

Now, let us assume that the covariance matrix of the mixtures is modeled by
$\sum_{n=1}^Nh_{ijn}\mathbf{R}_{in}$ as in Fig.~\ref{fig:param}, and that
both the JD and the rank-1 constraints are imposed on $\mathbf{R}_{in}$.\footnote{It is also possible to model the covariance matrix of the mixtures by $\sum_{n=1}^Nh_{ijn}\mathbf{a}_{in}\mathbf{a}_{in}^H+\mathbf{B}_i$ with $\mathbf{B}_i$ being the noise covariance matrix and to assume that $\mathbf{a}_{i1}\mathbf{a}_{i1}^H,\ldots,\mathbf{a}_{iN}\mathbf{a}_{iN}^H,\mathbf{B}_i$ are jointly diagonalizable. We do not discuss this case any further in this paper, and 
 include it in the future work.}
Let us consider the following cases:
\begin{itemize}\itemsep 0pt
\itemsep 0pt
\item Case $N>M$: In this case, Theorem~\ref{thm:shrink} implies $\mathbf{R}_{in}\propto\mathbf{R}_{in^\prime}$ for some distinct $n$ and $n^\prime$.
\item Case $N<M$: In this case,
$\sum_{n=1}^Nh_{ijn}\mathbf{R}_{in}$ is rank-deficient,
and so
 the negative log-likelihood is not well-defined.
\end{itemize}
Since both cases must be avoided, $N=M$. 
Now, let us write $\mathbf{R}_{in}=\mathbf{a}_{in}\mathbf{a}_{in}^H$.
Let us define a separation matrix
 $\mathbf{W}_i^H\coloneqq (\mathbf{a}_{i1},\ldots,\mathbf{a}_{iN})^{-1}$ (we assume that the mixing matrix $(\mathbf{a}_{i1},\ldots,\mathbf{a}_{iN})$ is nonsingular).
Then, $\mathbf{a}_{in}=\mathbf{W}_i^{-H}\mathbf{e}_n$, yielding
$\mathbf{R}_{in}=\mathbf{W}_i^{-H}\mathbf{e}_n\mathbf{e}_n^T\mathbf{W}_i^{-1}$,
the parametrization in time-varying Gaussian ICA in Fig.~\ref{fig:param}.
By plugging this in (\ref{eq:mwf2}), we have
\begin{align}
&\mathbf{F}_{ijn}^{MWF}\notag\\
&=\underbrace{\mathbf{W}^{-H}_i}_{\text{projection back}}\ \ \underbrace{\mathbf{e}_{n}\mathbf{e}_{n}^T}_{\text{remove all but the $n$th separated signal}}\ \ \underbrace{\mathbf{W}^{H}_i.}_{\text{separation matrix}}
\end{align}

\section{Connection with Approximate Joint Diagonalization (AJD)}
\label{sec:conajd}

This section reveals a close connection between FastFCA and 
 approximate joint diagonalization (AJD)~\cite{Flury1986,Cardoso1993,Belouchrani1997,Pham2001,Veen2001,Yeredor2002,Alyani2017}.

\subsection{Review of AJD}
\label{sec:AJD}



As we alluded to earlier, for any two matrices $\mathbf{Q}_1,\mathbf{Q}_2\in S_{++}^M$,
there exists $\mathbf{W}\in GL(M)$ such that $\mathbf{W}^H\mathbf{Q}_1\mathbf{W},\mathbf{W}^H\mathbf{Q}_2\mathbf{W}$
are {\it exactly} jointly diagonal, which is available as a solution to a generalized eigenvalue problem~\cite{Horn1990}.
However, this is not the case for general matrices $\mathbf{Q}_1,\ldots,\mathbf{Q}_J\in \mathbb{C}^{M\times M}$,
which is why AJD comes into play.
Given $\mathbf{Q}_1,\ldots,\mathbf{Q}_J\in \mathbb{C}^{M\times M}$, 
AJD seeks for $\mathbf{W}\in GL(M)$ such that
 $\mathbf{W}^H\mathbf{Q}_1\mathbf{W},\ldots,\mathbf{W}^H\mathbf{Q}_J\mathbf{W}$
are {\it as jointly diagonal as possible}.

%


Here we focus on Flury's AJD cost~\cite{Flury1986} based on the log-determinant divergence, 
which is closely related to FastFCA.
It is given by
\begin{align}
&\sum_{j=1}^J
\alpha_jD_{LD}\bigl(\mathbf{W}^H\mathbf{Q}_j\mathbf{W}\mid\ddiag\bigl(\mathbf{W}^H\mathbf{Q}_j\mathbf{W}\bigr)\bigr)\label{eq:AJDLD}\\
&=\sum_{j=1}^J
\alpha_j\ln\frac{\det\bigl(\ddiag\bigl(\mathbf{W}^H\mathbf{Q}_j\mathbf{W}\bigr)\bigr)}{\det\bigl(\mathbf{W}^H\mathbf{Q}_j\mathbf{W}\bigr)},
\end{align}
where $\mathbf{Q}_1,\ldots,\mathbf{Q}_J\in S^M_{++}$,
$\alpha_j>0$ is a weight, and $\ddiag(\cdot)$ nullifies the off-diagonal elements.
Although Flury {\it et al.}~\cite{Flury1986} restricted $\mathbf{W}$ to orthogonal matrices,
we drop this restriction as in~\cite{Pham2001}.
$D_{LD}$ is
the log-determinant divergence defined for
 $\bm{\Omega}_1, \bm{\Omega}_2\in S^M_{++}$ 
by
$
D_{LD}(\bm{\Omega}_1\mid\bm{\Omega}_2)\coloneqq \tr\bigl(\bm{\Omega}_1\bm{\Omega}_2^{-1}\bigr)-\ln\det\bigl(\bm{\Omega}_1\bm{\Omega}_2^{-1}\bigr)-M.
$
It is 
nothing but the Kullback-Leibler divergence between $\mathcal{N}_c(\cdot\mid\mathbf{0},\bm{\Omega}_1)$ and $\mathcal{N}_c(\cdot\mid\mathbf{0},\bm{\Omega}_2)$,
and hence
$D_{LD}(\bm{\Omega}_1\mid\bm{\Omega}_2)\geq 0$ with equality if and only if $\bm{\Omega}_1=\bm{\Omega}_2$ (see also Theorem~\ref{thm:LD} in Appendix B).
Consequently, $(\ref{eq:AJDLD})\geq 0$ with equality if and only if $\mathbf{W}^H\mathbf{Q}_1\mathbf{W},\ldots,\mathbf{W}^H\mathbf{Q}_J\mathbf{W}$
are exactly jointly diagonal.

\subsection{Piecewise Stationary Extension of FastFCA}

Here we extend the FastFCA model to piecewise stationary one in line with~\cite{Pham2001,Sawada2018}.
Let us partition the frames into $J$ disjoint subsets called {\it blocks},
each consisting of $B$ consecutive frames, where the $j$th block consists of 
the $((j-1)B+1)$th to the $(jB)$th frames.
Let $b=1,\ldots,B$ be the frame index within each block, and
$\mathbf{x}_{j}^{(b)}\in\mathbb{C}^M$
the mixtures
in the $b$th frame of the $j$th block.
We model it by $\mathbf{x}_{j}^{(b)}=\sum_{n=1}^N\mathbf{c}_{jn}^{(b)}$ with $\mathbf{c}_{jn}^{(b)}$ being the $n$th source image,
and assume that the $N$ source images independently follow 
$p(\mathbf{c}_{jn}^{(b)}\mid\Theta)
=\mathcal{N}_c(\mathbf{c}_{jn}^{(b)}\mid \mathbf{0},h_{ijn}\mathbf{W}^{-H}\bm{\Lambda}_{n}\mathbf{W}^{-1}).$
Note that this distribution  is independent of  $b$, and
hence it is said to be piecewise stationary.
Consequently, on the assumption of sample independence,
a normalized negative log-likelihood to be minimized is obtained as follows:
\begin{align}
&-\frac{1}{B}\ln p\Bigl(\mathbf{x}_{1:J}^{(1:B)}\mid\Theta\Bigr)\notag\\
&\overset{c}{=}-J\ln|\det\mathbf{W}|^2+\sum_{mj}\frac{\langle|y_{mj}|^2\rangle}{\sigma^2_{mj}}
+\sum_{mj}\ln\sigma^2_{mj}\label{eq:negllpws}
\end{align}
in much the same way as (\ref{eq:negll}).
Here, 
$y_{mj}^{(b)}\coloneqq\mathbf{w}_{m}^H\mathbf{x}_{j}^{(b)}$ denotes the $m$th decorrelated mixture, 
$\langle|y_{mj}|^2\rangle\coloneqq \frac{1}{B}\sum_{b=1}^B|y_{mj}^{(b)}|^2$  its short-term power in the $j$th block,
and $\sigma_{mj}^2$ is defined by (\ref{eq:sigmadef}).
Note that (\ref{eq:negllpws}) reduces to the previous FastFCA cost when $B=1$.
Algorithms~\ref{algo:FastFCA} and \ref{algo:FastFCAMM}
 can be used as they are for minimizing (\ref{eq:negllpws}), where the input 
$\widehat{\mathbf{X}}_{j}$ is given by $\widehat{\mathbf{X}}_{j}\coloneqq\langle\mathbf{x}_{j}\mathbf{x}_{j}^H\rangle
= \frac{1}{B}\sum_{b=1}^B \mathbf{x}_{j}^{(b)}\mathbf{x}_{j}^{(b)H}$, instead of (\ref{eq:Xhatdef}).

\subsection{Connection between FastFCA and AJD}
\label{sec:AJDFastFCA}
Now we are ready to show the connection between FastFCA and AJD.
%
(\ref{eq:negllpws}) can be rewritten as follows\footnote{We assume 
$\widehat{\mathbf{X}}_{j}\in S_{++}^M$, 
which is always true in practice provided $B\geq M$. }:
\begin{align}
(\ref{eq:negllpws})&=\sum_{j=1}^J
D_{LD}\Bigl(\mathbf{W}^H\widehat{\mathbf{X}}_{j}\mathbf{W}\mid \ddiag\Bigl(\mathbf{W}^H\widehat{\mathbf{X}}_{j}\mathbf{W}\Bigr)\Bigr)\notag\\
&\phantom{=}+\sum_{m=1}^M\sum_{j=1}^JD_{IS}(\langle |y_{mj}|^2\rangle\mid\sigma^2_{mj}).\label{eq:ldpytha}
\end{align}
%
%
The first term of (\ref{eq:ldpytha}) is an AJD cost, 
which is  (\ref{eq:AJDLD})  with $\mathbf{Q}_j=\widehat{\mathbf{X}}_{j}$ and $\alpha_j=1$, {\it i.e.,} the
cost of time-varying Gaussian ICA\footnote{In the piecewise stationary extension of 
time-varying Gaussian ICA~\cite{Pham2001},
the cost  can also be written as (\ref{eq:ldpytha}),
where $\sigma_{mj}^2$ is a free parameter, instead of (\ref{eq:sigmadef}).
In this case, $\sigma_{mj}^2$ can be eliminated from (\ref{eq:ldpytha}) by 
substituting the minimizer with respect to $\sigma_{mj}^2$:
 $\sigma_{mj}^2=\langle |y_{mj}|^2\rangle$.
This makes the second term in (\ref{eq:ldpytha}) vanish,
whereby remains only the first term,
the cost for AJD.
This result was essentially obtained in~\cite{Pham2001}.}.
The second term is an NMF cost.
Therefore,  FastFCA  can be regarded as a  regularized version of  AJD,
where the regularization term is NMF.
It is interesting to note that two different types of joint diagonalization have been encountered in FastFCA.
Indeed, we started with the JD constraint on the SCMs of all $N$ sources, 
but finally ended up with AJD of the short-term sample covariance matrices in all $J$ blocks.

{
\section{Performance Evaluation}
\label{sec:exp}
To confirm the effectiveness of methods presented in this paper, we conducted experiments. 
As in Fig.~\ref{fig:config}, we measured impulse responses in a real room.
There were $M=3$ microphones at the vertices of an equilateral triangle of side $4$\,cm.
There were $N=2$, $3$, or $4$ sources,
corresponding to the overdetermined, the determined, or the underdetermined case, respectively.
The sources at $70^\circ$ and $150^\circ$ were used when $N=2$, and
those at $70^\circ$, $150^\circ$, and $245^\circ$ when $N=3$.
The reverberation time $RT_{60}$ was varied 
by detaching some of the cushion walls so that $RT_{60}=130, 200, 250, 300, 370$, or $440$\,ms.
The source images were generated by convolving 8-s English speech source signals with impulse responses. 
The mixtures were generated by adding these source images. 
Ten BSS trials with different source signals
were carried out for each number of sources and for each reverberation time.
The number of time-frequency points was $IJ=$127,737.
The other conditions are summarized in Table~\ref{table:cond}.

\begin{table}
\centering
\caption{
Experimental conditions.
}
\label{table:cond}
\begin{tabular}{ll}\hline
sampling frequency&16\,kHz\\
frame length&1024 points (64\,ms)\\
frame shift&512 points (32\,ms)\\
window function&square root of Hann\\
number of EM/MM iterations&20\\\hline
\end{tabular}
\end{table}

FCA and FastFCA were implemented in Matlab (R2013a)\footnote{Sample codes for FastFCA are available online at \url{https://github.com/nttcslab-sp/unifiedUdetDetBSS/}.}
 on an Intel i7-2600 quad-core CPU with a base operation frequency of 3.4\,GHz.
In FastFCA, parameters were initialized by the following procedure:
\begin{enumerate}
\item \textbf{Estimate} $\mathbf{R}_{in}$\textbf{: }Time-frequency masks $\mathcal{M}_{ijn}$ without permutation ambiguity were estimated by Sawada {\it et al.}'s clustering-based BSS method~\cite{Sawada2011TASLP},
and $\mathbf{R}_{in}$ by $\mathbf{R}_{in}\leftarrow\frac{1}{J}\sum_{j=1}^J\mathcal{M}_{ijn}\mathbf{x}_{ij}\mathbf{x}_{ij}^H$ as in~\cite{Souden2013}. 
\item \textbf{Initialize} $\mathbf{W}_i$\textbf{: }$\mathbf{W}_i$ was
initialized by the solution $\mathbf{W}$ to the generalized eigenvalue problem $\mathbf{R}_{i2}\mathbf{W}=\mathbf{R}_{i1}\mathbf{W}\bm{\Delta}$ with $\bm{\Delta}$ being the diagonal matrix composed of the generalized eigenvalues.
\item \textbf{Initialize} $\bm{\Lambda}_{in}$\textbf{: }$\bm{\Lambda}_{in}\leftarrow \ddiag(\mathbf{W}_i^H\mathbf{R}_{in}\mathbf{W}_i)$.
\item \textbf{Initialize} $h_{ijn}$\textbf{: }$h_{ijn}\leftarrow\mathcal{M}_{ijn}\frac{1}{M}\mathbf{x}_{ij}^H\mathbf{R}_{in}^{-1}\mathbf{x}_{ij}$.
\end{enumerate}
In FCA, parameters were initialized by steps 1 and 4 only.

\begin{figure}
\centering
\includegraphics[width=0.8\columnwidth]{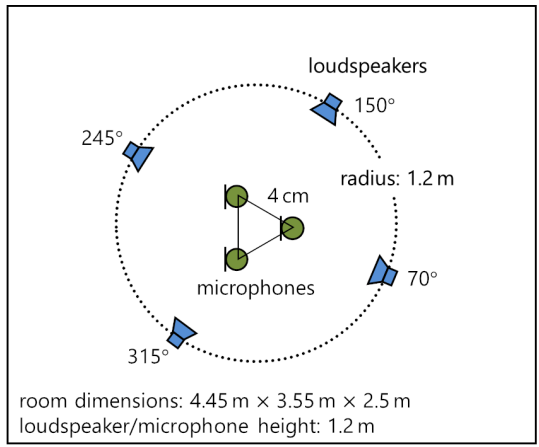}\\
\caption{Room configuration (plan view) in the experiments.}
\label{fig:config}
\end{figure}

Fig.~\ref{fig:RTF} compares FCA and FastFCA in terms of the computational cost measured by the real time factor (RTF), which is defined by the computational time divided by the data length. 
The RTF was averaged over all six reverberation times and ten trials.
In FCA, matrix inversion was performed by the Matlab slash `/'.
We clearly see that FastFCA gave
a significantly lower RTF than FCA.

Figs.~\ref{fig:overdetermined}, \ref{fig:determined}, and \ref{fig:underdetermined} compare the methods in terms of source separation performance measured by the signal-to-distortion ratio (SDR)~\cite{Vincent2006} for $N=2,3$, and $4$, respectively.
The SDR was averaged over all $N$ sources and ten trials.
FastFCA was able to give an SDR comparable to FCA for both the EM and the MM-based algorithms. 
It is interesting to note that, when $RT_{60}=130$\,ms and $N=2,3$,
the MM-based algorithm gave a significantly higher SDR
than the EM algorithm for both FCA and FastFCA.
\begin{figure}
\centering
\includegraphics[width=\columnwidth]{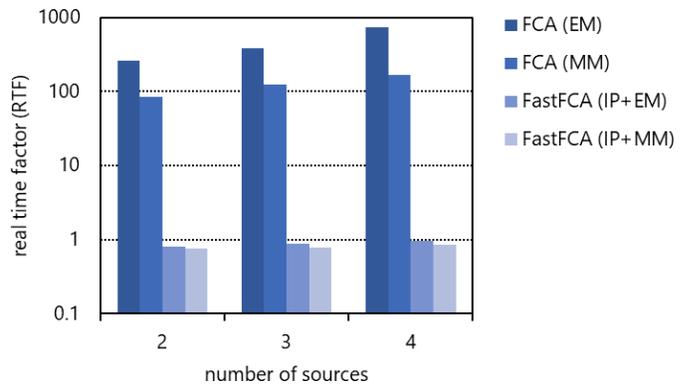}\\
\caption{Computational cost measured by the RTF.}
\label{fig:RTF}
\end{figure}
\begin{figure}
\centering
\includegraphics[width=\columnwidth]{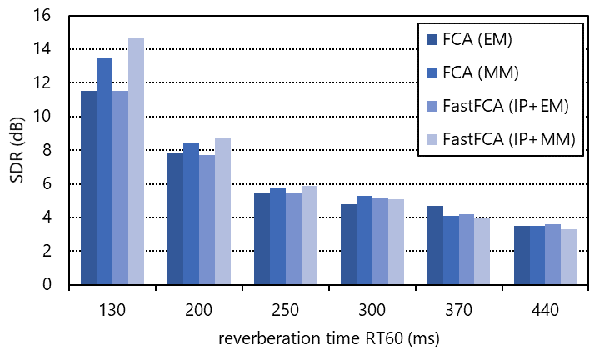}\\
\caption{Source separation performance measured by the SDR for $N=2$ (overdetermined case).}
\label{fig:overdetermined}
\end{figure}
\begin{figure}
\centering
\includegraphics[width=\columnwidth]{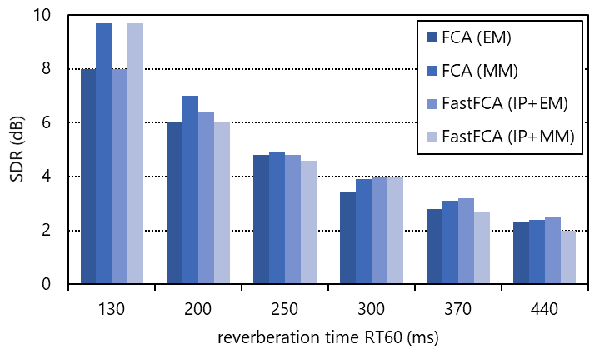}\\
\caption{Source separation performance measured by the SDR for $N=3$ (determined case).}
\label{fig:determined}
\end{figure}

\begin{figure}
\centering
\includegraphics[width=\columnwidth]{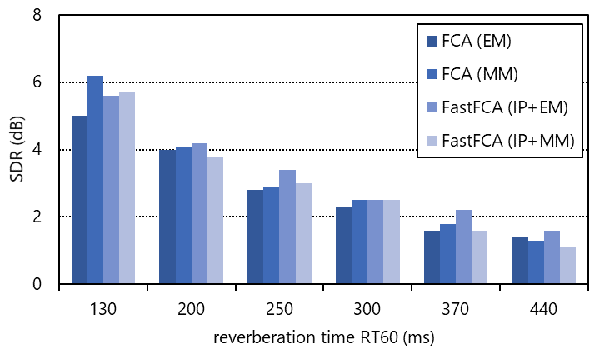}\\
\caption{Source separation performance measured by the SDR for $N=4$ (underdetermined case).}
\label{fig:underdetermined}
\end{figure}

To confirm the validity of the JD constraint, we compared FCA and FastFCA in terms of  the estimation error of the SCM.
The error was measured by the mean square error
between the estimate $\widehat{\mathbf{R}}_{in}$ and the ground truth $\mathbf{R}_{in}$,
\begin{align}
\frac{1}{IN}\sum_{i=1}^I\sum_{n=1}^N\Bigl\|\widehat{\mathbf{R}}_{in}-\mathbf{R}_{in}\Bigr\|_F^2,\label{eq:Rerror}
\end{align}
where $\widehat{\mathbf{R}}_{in}$ and $\mathbf{R}_{in}$ are scaled so that $\tr(\widehat{\mathbf{R}}_{in})=\tr(\mathbf{R}_{in})=1$
and $\|\cdot\|_F$ denotes the Frobenius norm.
In FCA, $\widehat{\mathbf{R}}_{in}$ was obtained by the EM or the MM algorithm.
In FastFCA, $\widehat{\mathbf{R}}_{in}$ was computed by (\ref{eq:parametrizationFastFCA}) from $\mathbf{W}_i$ and $\bm{\Lambda}_{in}$ estimated by the IP+EM or the IP+MM algorithm.
The ground truth $\mathbf{R}_{in}$ was obtained by ML estimation from the source image $\mathbf{c}_{ijn}$.
Table~\ref{table:error} shows the error (\ref{eq:Rerror}), averaged over all ten trials, all six reverberation times, and all $N=2,3,4$.
As seen from the table, FastFCA gave a slightly smaller error than FCA in both the EM and the MM cases.
FastFCA (IP+EM) gave the smallest error of $0.528$
with a $6.2\,\%$ relative reduction compared to the second smallest error of $0.563$ by FCA (EM).
This is likely because FCA has quite many unknown parameters ($IJN$ real numbers + $IN$ Hermitian matrices of order $M$) compared to the available data ($MIJ$ complex STFT coefficients), 
which makes parameter estimation difficult.
FastFCA can mitigate this issue thanks to the JD constraint.
This result clearly shows the validity of the JD constraint.

\begin{table}
\centering
\caption{
Estimation error of $\mathbf{R}_{in}$  measured by (\ref{eq:Rerror}).
}
\label{table:error}
\begin{tabular}{cccc}\hline
FCA (EM)&FCA (MM)&FastFCA (IP+EM)&FastFCA (IP+MM)\\
0.563&
0.634&
$\bm{0.528}$&
0.628\\\hline
\end{tabular}
\end{table}

\section{Conclusion}
\label{sec:conc}
We presented FastFCA, an efficient method for underdetermined audio BSS based on joint diagonalization. We also presented a unified framework for underdetermined and determined audio BSS, which highlights a close theoretical connection beteween FastFCA and other methods. Moreover, we revealed a connection between FastFCA and AJD. The future work includes extensions of FastFCA that deal with initialization robustness, convergence acceleration, and non-Gaussian distributions/divergences.


\appendices
\section{Derivation of the EM Algorithm for FCA}
On the assumption of sample independence and source independence,
$
p(Z\mid\Theta^\prime)
=\prod_{j=1}^J\prod_{n=1}^{N-1}\mathcal{N}_c(\mathbf{c}_{jn}\mid \mathbf{0},h_{jn}^\prime\mathbf{R}_n^\prime)
=\prod_{j=1}^J\mathcal{N}_c\bigl(\widetilde{\mathbf{c}}_{j}\mid \mathbf{0},\widetilde{\mathbf{C}}_j(\Theta^\prime)\bigr),
$
where $\widetilde{\mathbf{c}}_j\coloneqq (\mathbf{c}_{j1}^T,\ldots,\mathbf{c}_{j,N-1}^T)^T$. 
Here, $\widetilde{\mathbf{C}}_j(\Theta^\prime)$ denotes the block diagonal matrix composed of $h_{j1}^\prime\mathbf{R}_1^\prime,\ldots,
h_{j,N-1}^\prime\mathbf{R}_{N-1}^\prime$.
We also have
$
p(X\mid Z,\Theta^\prime)
=\prod_{j=1}^J\mathcal{N}_c(\mathbf{x}_j\mid\mathbf{E}\widetilde{\mathbf{c}}_{j},h_{jN}^\prime\mathbf{R}_N^\prime)
$
with $\mathbf{E}\coloneqq (\underbrace{\mathbf{I},\ldots,\mathbf{I}}_{N-1\text{ times}})$.
From the Bayes theorem, 
$p(Z\mid X,\Theta^\prime)=\prod_{j=1}^J\mathcal{N}_c\bigl(\widetilde{\mathbf{c}}_j\mid \bm{\mu}^{\widetilde{c}}_j,\bm{\Sigma}^{\widetilde{c}}_j\bigr)$
with
$\bm{\mu}^{\widetilde{c}}_j\coloneqq\widetilde{\mathbf{C}}_j(\Theta^\prime)\mathbf{E}^T\mathbf{X}_j(\Theta^\prime)^{-1}\mathbf{x}_j$ and 
$\bm{\Sigma}^{\widetilde{c}}_j\coloneqq\widetilde{\mathbf{C}}_j(\Theta^\prime)-\widetilde{\mathbf{C}}_j(\Theta^\prime)\mathbf{E}^T\mathbf{X}_j(\Theta^\prime)^{-1}\mathbf{E}\widetilde{\mathbf{C}}_j(\Theta^\prime).$
Thus,
\begin{align}
&Q(\Theta\mid\Theta^\prime)\notag\\
&\overset{c}{=}-\sum_{j=1}^J\sum_{n=1}^N\ln\det(h_{jn}\mathbf{R}_n)\\
&\notag\phantom{=}-\sum_{j=1}^J\tr\Bigl\{\widetilde{\mathbf{C}}_j(\Theta)^{-1}\Bigl[
\bm{\mu}^{\widetilde{c}}_j\Bigl(\bm{\mu}^{\widetilde{c}}_j\Bigr)^H+\bm{\Sigma}^{\widetilde{c}}_j\Bigr]\Bigr\}\\
&\notag\phantom{=}-\sum_{j=1}^J\tr\Bigl\{(h_{jN}\mathbf{R}_N)^{-1}\Bigl[\Bigl(\mathbf{x}_j-\mathbf{E}\bm{\mu}^{\widetilde{c}}_j\Bigr)\Bigl(\ldots\Bigr)^H
+\mathbf{E}\bm{\Sigma}^{\widetilde{c}}_j\mathbf{E}^T\Bigr]\Bigr\}\\
&=-\sum_{j=1}^J\sum_{n=1}^N\ln\det(h_{jn}\mathbf{R}_n)
-\sum_{j=1}^J\sum_{n=1}^N\tr[(h_{jn}\mathbf{R}_n)^{-1}\bm{\Psi}_{jn}
],
\end{align}
where $\bigl(\ldots\bigr)^H$ denotes $\bigl(\mathbf{x}_j-\mathbf{E}\bm{\mu}^{\widetilde{c}}_j\bigr)^H$ and $\bm{\Psi}_{jn}\coloneqq
\bm{\mu}^c_{jn}(\bm{\mu}^c_{jn})^H+
\bm{\Sigma}^c_{jn}$.
Here, $
\bm{\mu}^c_{jn}\coloneqq h_{jn}^\prime\mathbf{R}_n^\prime\mathbf{X}_j(\Theta^\prime)^{-1}\mathbf{x}_j$ and $
\bm{\Sigma}^c_{jn}\coloneqq h_{jn}^\prime\mathbf{R}_n^\prime-(h_{jn}^\prime)^2\mathbf{R}_n^\prime\mathbf{X}_j(\Theta^\prime)^{-1}\mathbf{R}_n^\prime$ are
 the parameters of  the marginal posterior of each source image,
$
p(\mathbf{c}_{jn}\mid\mathbf{x}_j,\Theta^\prime)=\mathcal{N}_c(\mathbf{c}_{jn}\mid \bm{\mu}^c_{jn},\bm{\Sigma}^c_{jn}).
$
Therefore, in the E-step, it suffices to update $\bm{\mu}^c_{jn}\in \mathbb{C}^M$ and $\bm{\Sigma}^c_{jn}\in S_{++}^M$,
instead of $\bm{\mu}^{\widetilde{c}}_j\in \mathbb{C}^{M(N-1)}$ and $\bm{\Sigma}^{\widetilde{c}}_j\in S_{++}^{M(N-1)}$.
The update rules  for the M-step in (\ref{eq:FCAEM2}) are obtained by partial differentiation of $Q$.

\section{Derivation of the MM Algorithm for FCA}
First of all, we have the following theorems:
\begin{theorem}
\label{thm:LD}
Let $M$ be a positive integer.
For $\bm{\Omega}_1,\bm{\Omega}_2\in S^M_{++}$, $\tr(\bm{\Omega}_1\bm{\Omega}_2^{-1})-\ln\det(\bm{\Omega}_1\bm{\Omega}_2^{-1})-M\geq 0$ with equality if and only if $\bm{\Omega}_1=\bm{\Omega}_2$.
\end{theorem}
\begin{proof}
There exist $\mathbf{W}\in GL(M)$ and diagonal $\bm{\Lambda}\in S_{++}^M$ such that $\mathbf{W}^H\bm{\Omega}_1\mathbf{W}=\bm{\Lambda}$ and $\mathbf{W}^H\bm{\Omega}_2\mathbf{W}=\mathbf{I}$~\cite{Horn1990}. Therefore, $\tr(\bm{\Omega}_1\bm{\Omega}_2^{-1})-\ln\det(\bm{\Omega}_1\bm{\Omega}_2^{-1})-M=\tr(\bm{\Lambda})-\ln\det(\bm{\Lambda})-M=\sum_{m=1}^M
([\bm{\Lambda}]_{mm}-\ln[\bm{\Lambda}]_{mm}-1)\geq 0$. The equality condition is obvious.
\end{proof}
\begin{theorem}\label{thm:jensen}
Let $M,N$ be positive integers.
Suppose $\bm{\Omega}_1,\ldots,\bm{\Omega}_N\in S^M_{++}$, $\bm{\Gamma}_1,\ldots,\bm{\Gamma}_N\in\mathbb{C}^{M\times M}$, and
$\sum_{n=1}^N\bm{\Gamma}_n=\mathbf{I}$. We have $\sum_{n=1}^N\bm{\Gamma}_n^H\bm{\Omega}_n^{-1}\bm{\Gamma}_n\geq 
(\sum_{n=1}^N\bm{\Omega}_n)^{-1}$ with equality if and only if, for all $n\in\{1,\ldots,N\}$, $\bm{\Gamma}_n=\bm{\Omega}_n(\sum_{\nu=1}^N\bm{\Omega}_\nu)^{-1}$.
Here, for Hermitian matrices $\mathbf{X},\mathbf{Y}$ of the same size, $\mathbf{X}\geq \mathbf{Y}\overset{def}{\Longleftrightarrow}\mathbf{X}-\mathbf{Y}$ is positive semidefinite.
\end{theorem}
\begin{proof}
$\sum_{n=1}^N\bm{\Gamma}_n^H\bm{\Omega}_n^{-1}\bm{\Gamma}_n-
(\sum_{n=1}^N\bm{\Omega}_n)^{-1}=\sum_{n=1}^N[\bm{\Gamma}_n-\bm{\Omega}_n(\sum_{\nu=1}^N\bm{\Omega}_\nu)^{-1}]^H\bm{\Omega}_n^{-1}[\bm{\Gamma}_n-\bm{\Omega}_n(\sum_{\nu=1}^N\bm{\Omega}_\nu)^{-1}]\geq \mathbf{O}.$ The equality condition is obvious.
\end{proof}

The above theorems imply that $\mathcal{J}$ in (\ref{eq:negllFCA}) is majorized by an auxiliary function $\mathcal{J}^+$ as follows:
\begin{align}
\mathcal{J}(\Theta)&=\sum_{j=1}^J[\ln\det \mathbf{X}_j(\Theta)-\ln g(\mathbf{x}_j^H\mathbf{X}_j(\Theta)^{-1}\mathbf{x}_j)]\\
&\leq \sum_{j=1}^J\Biggl[\ln\det\bm{\Pi}_j+\sum_{n=1}^Nh_{jn}\tr(\mathbf{R}_n\bm{\Pi}_j^{-1})-M\\
\notag &\phantom{\coloneqq}-\ln g\Biggl(\mathbf{x}_j^H\sum_{n=1}^N\frac{\bm{\Gamma}^H_{jn}\mathbf{R}_n^{-1}\bm{\Gamma}_{jn}}{h_{jn}}\mathbf{x}_j\Biggr)\Biggr]\\
&=:\mathcal{J}^+(\Theta,\Xi).
\end{align} 
Here, $\Xi\coloneqq \{\bm{\Pi}_{1:J},\bm{\Gamma}_{1:J,1:N}\}$ is the set of auxiliary variables, and\footnote{Obviously, the above inequality remains valid for any nonincreasing function $g: \mathbb{R}_+\rightarrow \mathbb{R}_{++}$ ($\mathbb{R}_{+}$: the set of non-negative numbers), corresponding to a subclass of a general complex elliptically symmetric distribution~\cite{Ollila2012}.} $g(\alpha)\coloneqq\exp(-\alpha)$.
Update rules of $\Xi$ can be obtained from the equality conditions in Theorems~\ref{thm:LD} and \ref{thm:jensen} as 
$\bm{\Gamma}_{jn}\leftarrow h_{jn}\mathbf{R}_n(\sum_{\nu=1}^Nh_{j\nu}\mathbf{R}_\nu)^{-1}$ and
$\bm{\Pi}_j\leftarrow\sum_{n=1}^Nh_{jn}\mathbf{R}_n$.
Update rules of $\Theta$ can be obtained by partial differentiation of  $\mathcal{J}^+$ as
$h_{jn}\leftarrow\sqrt{\frac{\mathbf{x}_j^H\bm{\Gamma}_{jn}^H\mathbf{R}_n^{-1}\bm{\Gamma}_{jn}\mathbf{x}_j}{\tr(\mathbf{R}_n\bm{\Pi}_j^{-1})}}$
and $\mathbf{R}_n\leftarrow(\sum_{j=1}^Jh_{jn}\bm{\Pi}_j^{-1})^{-1}\#(\sum_{j=1}^J\frac{\bm{\Gamma}_{jn}\mathbf{x}_j\mathbf{x}_j^H\bm{\Gamma}_{jn}^H}{h_{jn}})$. By eliminating  $\bm{\Gamma}_{jn}$, we have the MM algorithm for FCA in (\ref{eq:updateXMM})--(\ref{eq:updateRMM}).

%


\section{FastMNMF}
\label{sec:fastmnmf}

This appendix describes FastMNMF.
The spectrogram of each source signal is often composed of a few recurrent spectral patterns,
{\it e.g.}, musical notes played by an instrument. This motivates us to
 model the spectrogram of the $n$th source signal, $(h_{ijn})_{ij}$, by the product of two non-negative
matrices: 
$\mathbf{T}_n\mathbf{V}_n.$
Here, $\mathbf{T}_n$ is composed of
columns that model
the recurrent spectral patterns, and $\mathbf{V}_n$ rows that model their temporal activations.
In scalar form, we can write
\begin{align}
h_{ijn}=\sum_{k=1}^Kt_{ikn}v_{jkn},\label{eq:NMFconst2}
\end{align}
where $t_{ikn}\coloneqq [\mathbf{T}_n]_{ik}$ and $v_{jkn}\coloneqq [\mathbf{V}_n]_{kj}$.

In FastMNMF, each source image is modeled by (\ref{eq:FastFCAstosigmodel}) as in FastFCA, but
 $h_{ijn}$ is not a free parameter any more but parametrized as in (\ref{eq:NMFconst2}).
As in FastFCA, the negative log-likelihood to be minimized is given by
\begin{align}
&-\ln p(\mathbf{x}_{1:I,1:J}\mid\Theta)\notag\\
&\overset{c}{=}-J\sum_{i=1}^I\ln|\det\mathbf{W}_i|^2+\sum_{mij}\frac{|y_{mij}|^2}{\sigma^2_{mij}}
+\sum_{mij}\ln\sigma^2_{mij}.\label{eq:FastMNMFcost}
\end{align}
Here, $\sigma^2_{mij}$ is given by (\ref{eq:sigmadef}) with $h_{ijn}$ parametrized as in (\ref{eq:NMFconst2}), 
$\Theta\coloneqq\{\mathbf{W}_{1:I},\bm{\Lambda}_{1:I,1:N}, \mathbf{T}_{1:N}, \mathbf{V}_{1:N}\}$,
and $\sum_{mij}$ is a shorthand notation for $\sum_{m=1}^M\sum_{i=1}^I\sum_{j=1}^J$.
The first two terms of (\ref{eq:FastMNMFcost}) can be regarded as a cost of time-varying Gaussian ICA as before, which is relevant to updating $\mathbf{W}_i$. On the other hand, 
the last two terms of (\ref{eq:FastMNMFcost}) can be regarded as a cost of tensor decomposition based on the Itakura-Saito divergence, which is relevant to updating $\bm{\Lambda}_{in}$, $\mathbf{T}_{1:N}$, and $\mathbf{V}_{1:N}$. Consequently, we can minimize (\ref{eq:FastMNMFcost}) efficiently by alternately applying ICA updates for $\mathbf{W}_i$ and tensor decomposition updates for $\bm{\Lambda}_{in}$, $\mathbf{T}_n$, and $\mathbf{V}_n$. For example, $\mathbf{W}_i$ can be updated by IP,
and $\bm{\Lambda}_{in}$, $\mathbf{T}_n$, and $\mathbf{V}_n$ by the MM algorithm. 
See~\cite{Ito2019ICASSP} for more details.

\end{document}